%% file: main.tex
\documentclass[11pt]{article}
\usepackage{mystyle}

\newcommand{\qnorm}[2][q]{\|#2\|_{#1}}

\newcommand{\cost}{\mathrm{SC}}
\newcommand{\M}{\mathcal{M}}
\newcommand{\A}{\mathcal A}
\renewcommand{\phi}{\varphi}

\begin{document}
\title{Strategic Facility Location with $p$-Norm Social Costs}
\author{Jabari Hastings\thanks{Supported by the Simons Foundation Collaboration on the Theory of Algorithmic Fairness and the Simons Foundation Investigators Award 17351.}\\
Stanford University}

\date{}
\maketitle

\begin{abstract}
We consider the strategic facility location problem in $\ell_q(\mathbb R^d)$ spaces where the social cost is defined by an arbitrary $p$-norm of the individual costs. While the optimal approximation ratios for deterministic strategyproof mechanisms are well established in the $d = 1$ setting, the guarantees for multi-dimensional spaces under an arbitrary $p$-norm are less understood.

In this work, we analyze the well-studied, strategyproof coordinate-wise median (CM) mechanism and provide approximation guarantees for these generalized social costs.
\begin{itemize}
\item We show that the CM mechanism is in fact robust to a broader class of social objectives: for every monotone symmetric norm objective, including all $p$-norm social costs, its approximation ratio never exceeds $3$ in arbitrary $\ell_q(\mathbb R^d)$ spaces, regardless of the dimension.

\item For $d = 2$, we establish tight approximation ratios for all $p, q \geq 1$ in $\ell_q(\mathbb R^2)$. In particular, we show that the CM mechanism is a $2^{1-1/\max(p,q)}$-approximation, resolving the conjecture of Goel and Hann-Caruthers (Social Choice and Welfare, 2023) in the Euclidean case and extending the guarantee to arbitrary $\ell_q$ distances.

\item For $d\geq 3$, we refine the dimension-independent approximation guarantee for $p$-norm social costs in $\ell_q(\mathbb R^d)$ spaces, giving upper bounds that depend on the relationship between the social-cost norm $p$ and the underlying distance norm $q$. This generalizes the recent result of Gravin and Jia (STOC, 2025) for the utilitarian social cost.
\end{itemize}
\end{abstract}

\thispagestyle{empty}
\clearpage

\tableofcontents
\clearpage

\section{Introduction}
\input{intro}

\section{Preliminaries}
\input{prelims}

\section{Monotone Symmetric Norm Objectives in \texorpdfstring{$\R^d$}{Rd}}\label{sec:monotone-norm}
\input{monotone-norm}

\section{A Program for \texorpdfstring{$p$-Norm Social Costs}{p-Norm Social Costs}}\label{sec:program}
\input{program}

\section{\texorpdfstring{$p$-Norm Social Costs in $\ell_q(\R^2)$}{p-Norm Social Costs in lq(R2)}}\label{sec:plane}
\input{plane}

\section{\texorpdfstring{$p$-Norm Social Costs in $\ell_q(\R^d)$}{p-Norm Social Costs in lq(Rd)}}\label{sec:higher}
\input{higher-dimension}

\clearpage
\section{Conclusion}
\input{conclusion}

\clearpage
\bibliography{mybib}

\bibliographystyle{apalike}

\appendix

\section{Proofs from \texorpdfstring{\Cref{sec:higher}}{Higher Dimension}}
\input{higher-appendix}

\end{document}

%% file: intro.tex
Facility location is a fundamental problem at the intersection of combinatorial optimization and mechanism design. In its canonical form, there are $n$ agents positioned at locations $x_1, \dots, x_n$ within a metric space, and the goal is to select a facility location $f$ that (approximately) minimizes a social cost derived from the distances between the agents and the facility. Within the optimization community, research has focused extensively on generalizations of the problem, including multiple facilities \citep{hochbaum1985best,charikar1999constant,jain2001approximation} and a variety of cost functions \citep{chakrabarty_approximation_2019,abbasi_parameterized_2023}. We refer the reader to \citep{laporte2020introduction} for a comprehensive survey.

In contrast, the mechanism design community has primarily focused on the strategic component of the facility location problem. In this setting, a mechanism must operate on the agents' \emph{reported} locations, which may differ from their true locations. Since agents may strategically misreport their location to manipulate the outcome, the goal is to design \emph{strategyproof} mechanisms, which incentivize agents to truthfully report their locations. A long line of work aims to determine the approximation guarantees of these mechanisms \citep{procaccia2013approximate,FotakisT10,FotakisT14,Meir:2019aa,walsh2020strategy,Goel:2023aa,Gravin:2025aa}. The strategic facility location problem has also emerged as a sandbox for testing new concepts in mechanism design, including truthfulness without monetary payments \citep{procaccia2013approximate,SerafinoV16}
and mechanisms with predictions \citep{AgrawalBGTX22, XuL22, barak2024mac, Gravin:2025aa}.

While strategyproofness may seem like an innocuous constraint, it introduces unique challenges that do not arise in the standard algorithmic setting. Some of these challenges are already visible in closely related variants of the problem. For instance, in the two-facility setting, \citet{LuSWZ10} show that deterministic strategyproof mechanisms cannot achieve an approximation ratio better than $\Omega(n)$ for the sum-of-distances objective, even in the line metric. This contrasts with the non-strategic $k$-median setting, where deterministic algorithms have long been known to achieve constant-factor approximations even in general metric spaces \citep{charikar1999constant,jain2001approximation}. 
This contrast illustrates that strategyproofness can substantially change the approximation landscape of facility location problems.

Against this backdrop, we return to the single-facility setting in $\ell_q(\mathbb R^d)$, where each agent's individual cost is the distance from their location to the chosen facility. We focus on another crucial dimension of the facility location problem, namely the choice of social cost objective. Traditionally, assignment problems have focused on the \textit{utilitarian} social cost (i.e., the sum of individual costs), or the \textit{egalitarian} social cost (i.e., the maximum individual cost). While these objectives are compelling because they convey how the ``average'' agent or the ``worst-off'' agent is impacted, it may be desirable for a mechanism not to be overly sensitive to any particular objective.
An ideal mechanism would \textit{simultaneously} provide good approximation guarantees for a general collection of objectives (e.g., the class of $p$-norm social costs). 
This kind of guarantee asks whether the same mechanism remains effective across different ways of aggregating individual costs, thereby giving a more complete picture of how costs are distributed across agents.

Despite the extensive work on facility location over the last two decades, our understanding of the approximation guarantees of deterministic strategyproof mechanisms in $\ell_q(\mathbb R^d)$ spaces under $p$-norm social costs remains substantially limited. Below, we briefly review the current landscape.

\paragraph{The line ($d = 1$).}
This setting is well-understood.
\citet{Feigenbaum:2017aa} show that the \emph{median}
mechanism, a simple rule proposed by \citet{Black48} that returns the median reported location, provides a $2^{1 - 1/p}$ approximation ratio, and that this guarantee is optimal among all deterministic, strategyproof mechanisms.

\paragraph{The plane ($d = 2$).} As we increase the dimension, considerable gaps emerge. In $\ell_2(\mathbb R^2)$, the \emph{coordinate-wise} median (CM) mechanism, which applies the median mechanism on each coordinate, is known to have the optimal approximation ratio among all deterministic, anonymous, and strategyproof mechanisms. However, its exact approximation guarantee for an arbitrary $p$-norm social cost remains an open question. \citet{Goel:2023aa} show that this ratio lies in the interval $[2^{1 - 1/p}, 2^{3/2 - 2/p}]$ and conjectured that the true bound is at the lower end of this range.

\begin{conjecture}[\citealp{Goel:2023aa}]
In the Euclidean space $\ell_2(\mathbb R^2)$, under the $p$-norm social cost for any $p \ge 2$, the CM mechanism is a $2^{1 - 1/p}$-approximation.
\end{conjecture}

\paragraph{Higher dimensions ($d \ge 3$).}
The literature is even sparser. The only known approximation guarantees are limited to the utilitarian social cost ($p = 1$). A recent result by \citet{Gravin:2025aa} demonstrates that for this $p=1$ case in $\ell_q(\mathbb R^d)$, the approximation ratio of the CM mechanism is bounded by a value that approaches $3$ as $q \to \infty$. This result was surprising because it ran counter to a widely held belief that the mechanism has an approximation ratio of $\Omega(\sqrt{d})$ in $\mathbb R^d$. However, it remained unclear whether similar dimension-independent guarantees hold for other social objectives.
\medskip

As the preceding discussion illustrates,
the median mechanism serves as a natural benchmark for the class of deterministic strategyproof mechanisms for facility location.
Notably, in $\ell_2(\mathbb{R}^2)$, the mechanism is also uniquely characterized as the only deterministic, strategyproof mechanism that is both Pareto optimal and anonymous \citep{peters1993range}. Motivated by its fundamental role in the facility location setting, our work analyzes the approximation guarantees of the CM mechanism in general normed spaces $\ell_q(\mathbb{R}^d)$ under arbitrary $p$-norm social costs.

\subsection{Our Contributions}
We provide a significantly sharper understanding of the CM mechanism for general objectives.

\paragraph{Constant approximation ratio for monotone symmetric norms.}
We first show that the constant, dimension-independent approximation guarantee of \citet{Gravin:2025aa}
is not specific to the utilitarian objective. In fact, the CM mechanism has
a constant approximation ratio for every monotone symmetric norm applied to
the vector of individual costs, a class that includes all $p$-norm social costs.
\begin{restatable}{theorem}{thmmonotonesocialcost}\label{thm:monotone-social-cost}
In the normed space $\ell_q(\R^d)$ with $q\in[1,\infty]$ and $d\ge1$, for every $n\ge1$ and every monotone symmetric norm $\phi$ on $\R^n$, the CM mechanism has approximation ratio $\alpha_\phi(\M)\le3$. Moreover, when $d=2$, $\alpha_\phi(\M)\le2$.
\end{restatable}
Unlike the more involved analysis of \citet{Gravin:2025aa} for the utilitarian
case, which proceeds through a program relaxation and the solution of a system
of differential equations, our proof of \Cref{thm:monotone-social-cost} uses a
majorization argument based on standard consequences of Ky Fan dominance.
While simple, the argument perhaps unexpectedly reveals that the CM mechanism
controls not just the total cost, but also the cumulative cost borne by the
worst-off agents at various scales.

\paragraph{Tight $p$-norm approximation ratios in $\ell_q(\mathbb R^2)$.}
The preceding theorem gives a broad constant guarantee, but does not capture
the sharper dependence on $p$ and $q$. We next specialize to $p$-norm social
costs in the plane and determine the asymptotically tight approximation ratio
of the CM mechanism for every $p,q\ge1$.
Specifically, we show that the CM mechanism is a $2^{1-1/\max(p,q)}$-approximation, resolving the conjecture of \citet{Goel:2023aa} for the Euclidean plane and extending the approximation guarantees to all $\ell_q(\mathbb R^2)$ spaces.\footnote{Concurrently and independently, \citet{chan2026strategyproofmechanismseuclideanfacility} also resolved this conjecture for the Euclidean plane. We provide a detailed comparison of our concurrent results in \Cref{sec:related-work}.} We then complement these upper bounds with matching lower bounds.
\begin{restatable}{theorem}{thmupperboundtwo}\label{thm:cm-2d-upper-bound}
In the normed space $\ell_q(\R^2)$ with $p,q \geq 1$, the CM mechanism $\M$ has approximation ratio
\begin{equation*}
\alpha_p(\M) \leq
\begin{cases}
2^{1 - 1/p} & \text{if } 1 \leq q \leq p,  \\
2^{1 - 1/q} & \text{if } 1 \leq p < q.
\end{cases}
\end{equation*}
\end{restatable}

We derive our upper bounds on the approximation ratio by carefully analyzing the natural generalization of a mathematical program introduced by \citet{Gravin:2025aa} to study the CM mechanism for the utilitarian social cost ($p = 1$). While full details of the program appear in \Cref{sec:program}, we summarize the setup here and outline our refined analysis. At a high level, our goal is to minimize a quantity $\sum_{i \in [n]} g(x_i)$, where $g$ encodes the relative difference between an individual's distance to the optimal facility and their distance to the location chosen by the CM mechanism.

For the utilitarian objective, \citet{Gravin:2025aa} analyze a relaxation of the program that encodes the CM mechanism.
By deriving a uniform bound on the contribution of each $g(x_i)$, they establish an upper bound of $\sqrt{6 \sqrt{3} - 8} \approx 1.55$ for $\ell_2(\mathbb R^d)$. However, their approximation ratio is not tight in $\mathbb R^2$, exceeding the optimal guarantee of $\sqrt{2} \approx 1.41$ established by \citet{Meir:2019aa}.

To attain the exact $2^{1-1/\max(p,q)}$-approximation ratio in $\mathbb R^2$ for an arbitrary $p$-norm objective, we cannot afford to solve a relaxation of the program. Our improved analysis comes from two key observations.
First, we observe that the CM mechanism actually imposes a strict balancing constraint on the locations. In particular, when the output of the CM mechanism is the origin and $n$ is even, there must be an equal number of locations in opposite orthants. While this property trivially holds on the real line ($d=1$), our analysis shows that this phenomenon persists in two dimensions.
This balancing constraint leads us to a second, more crucial observation: we can uniformly lower bound the contribution of \emph{pairs} of locations, rather than a single location. When we take pairs of points $(a, b)$ in opposite orthants, the penalty of $g(a)$ can be offset by $g(b)$. We remark that \citet{Feigenbaum:2017aa} leveraged this second observation to determine the asymptotic ratio of the CM mechanism in $\mathbb R$ for all $p$-norm social costs.

With these two insights, our upper bound analysis of the CM mechanism becomes considerably streamlined, relying on standard inequalities. While such inequalities might suggest potential slack, the lower bounds in \Cref{thm:cm-2d-lower-bound} confirm that our analysis is indeed tight. For the case of $1 \leq q \leq p$, we use a two-agent construction and reason about strategic manipulation in a manner similar to that of \citet{Feigenbaum:2017aa,Goel:2023aa}. For the case of $1 \leq p < q$, we use a ``two clusters and a point'' instance (bearing some resemblance to a construction of \citet{Goel:2023aa}) that forces the CM mechanism to select the single point rather than the midpoint between the clusters.

\paragraph{Sharper $p$-norm guarantees in $\ell_q(\mathbb R^d)$.}
In the case of $d\geq 3$, the CM mechanism no longer enforces the strict balance constraint across opposite orthants observed in two dimensions, making the analysis significantly more involved. We therefore return to the optimization-program framework of \citet{Gravin:2025aa} and extend it to arbitrary $p$-norm social costs, obtaining dimension-independent bounds whose value depends on the relationship between $p$ and $q$.

\begin{restatable}{theorem}{thmupperboundhigh}\label{thm:cm-upper-bound}
In the normed space $\ell_q(\mathbb R^d)$ with $p,q \geq 1$ and $d \geq 1$, the CM mechanism $\M$ has approximation ratio $\alpha_p(\M) \leq \mathtt{UB}(p, q)$ where $\mathtt{UB}(p, q)$ is a function with the following properties:
\begin{enumerate}
    \item For every $p, q \geq 1$, if $p = q$ or $p = 2q$, then $\mathtt{UB}(p,q) = 2^{1 - 1/p}$.
    \item For every $p, q \geq 1$, if $p/q \in (1, 2)$, then $\mathtt{UB}(p, q) \le 1 + 2^{1/p}$.
    \item For every $p, q \geq 1$, if $p/q \in (0, 1) \cup (2, \infty)$, then $\mathtt{UB}(p, q) \leq 1 + 2^{1/\min{(p, q)}}$.
\end{enumerate}
In particular, $\alpha_p(\M) \leq 3$ for every $p, q \geq 1$.
\end{restatable}

Extending the analysis to an arbitrary $p$-norm social cost introduces technical challenges arising from the relationship between the cost function and $q$-norm distances---specifically, the subadditivity or superadditivity that occurs depending on the ratio $p/q$.

\subsection{Related Work}\label{sec:related-work}
The existing literature on facility location is extensive and the works mentioned in the introduction are by no means comprehensive. In addition to the textbook of \citet{laporte2020introduction}, we refer the reader to the surveys from the perspective of mechanism design \citep{chan2021mechanism} and social choice \citep{barbera1993generalized}. We also briefly highlight some closely related work.

\paragraph{Independent and Concurrent Work.}
The contemporaneous work of
\citet{chan2026strategyproofmechanismseuclideanfacility} also investigates strategic facility location under $p$-norm social costs in the Euclidean plane $\ell_2(\mathbb R^2)$. They independently resolve the conjecture of \citet{Goel:2023aa} and also give approximation guarantees for various randomized mechanisms. Their stated bounds for the CM mechanism are for the Euclidean plane, whereas our bounds apply in arbitrary $\ell_q(\mathbb R^2)$ spaces. They also do not consider the more general setting of $d \ge 3$.

A revised version of their work also proves that CM is a $2$-approximation for symmetric monotone norm objectives in the Euclidean plane. Prompted by this result and by feedback from anonymous reviewers, we considered the same broader objective class. In this direction, our main contribution is a dimension-independent $3$-approximation for CM under every monotone symmetric norm objective in arbitrary $\ell_q(\R^d)$. We also record a corresponding planar statement in $\ell_q(\R^2)$.

\paragraph{Strategyproof Mechanisms.}
The facility location problem can naturally be framed as an election.
Early works on voting with single-peaked preferences or their natural extensions include \citep{moulin1980strategy, barbera1993generalized,barbera1998strategy,schummer2002strategy,ching1997strategy}. Notably, \citet{kim1984nonmanipulability,peters1993range} study the voting problem when preferences are expressed in $\ell_2(\mathbb R^2)$, immediately translating to the setting we study. In fact, the prevailing characterization of deterministic strategyproof mechanisms for facility location in two dimensions also emerges from these works.
Providing a crisp specification for higher dimensions is an open problem.
Apart from the works of \citet{border1983straightforward,peremans1997strategy}, who consider classes of quadratic separable preferences in $\mathbb R^d$, our understanding of the general problem remains limited.

\paragraph{Randomized Facility Location.} Several works have investigated whether randomized mechanisms achieve improved approximation guarantees while preserving strategyproofness \citep{procaccia2013approximate,tang2020characterization,BalkanskiGS24}. Many of the proposed mechanisms are only strategyproof \textit{in expectation}. However, the recent work of \citet{barak2026facility} introduces a (universally) strategyproof randomized mechanism that attains an approximation ratio of $4/\pi \approx 1.27$ for the utilitarian objective in $\mathbb R^2$, which is strictly better than the CM mechanism.

\paragraph{Facility Location in Other Metrics.}
We note that there has been interest in studying facility location in other metric spaces. When the metric is a cycle, it is known that the only onto and strategyproof mechanism is a dictatorship \citep{schummer2002strategy}. Moreover, any deterministic strategyproof mechanism can have a worst-case approximation ratio of $\Omega(n)$~\citep{DokowFMN12}. The cycle results significantly limit the settings to trees \citep{DokowFMN12,FeldmanW13}, or to instances with a constant number of agents~\citep{Meir:2019aa}. We note that \citet{FilimonovM22} show that a variant of the median mechanism achieves optimal social cost on trees.

\paragraph{Additional Facility Location Variants.} Beyond the standard model, researchers have explored generalizations such as optimizing squared $\ell_2$ social costs \citep{FeldmanW13} and locating obnoxious facilities to maximize distance from agents \citep{cheng2013strategy, mei2019, li2023obnoxious, kanellopoulos2025constrained}. Other variants include dual preferences \citep{FeigenbaumS15, zou15}, concave costs \citep{FotakisT16}, distributed selection of multiple facilities \citep{filoratsikas2021}, and false-name-proofness \citep{todo11}, which extends strategyproofness to settings where agents can strategically inflate their influence by reporting multiple locations.

%% file: prelims.tex
\subsection{Notation and Terminology}
In the strategic facility location problem, the true locations of agents $[n]$ are denoted by the location profile $x = (x_1,x_2, \dots, x_n) \in (\R^d)^n$.
When a mechanism $\M$ takes these reported locations as input, it produces a facility location $\M(x) \in \R^d$.
For an agent located at $x_i$, the individual \emph{cost} of a facility location $f \in \R^d$ is the distance $\qnorm{f - x_i}$ for some $q \in [1, \infty]$. 

\begin{definition}
A mechanism $\M$ is (individually) \emph{strategyproof} if no agent can strictly decrease their cost by reporting a false location. Formally, for any agent $i \in [n]$ and any two location profiles $x, x' \in (\R^d)^n$ that differ only in the $i$-th agent's location (i.e., $x_j = x_j'$ for all $j \neq i$),
\begin{equation*}
    \qnorm{\M(x) - x_i} \le \qnorm{\M(x') - x_i}.
\end{equation*}
\end{definition}

If a mechanism is strategyproof, reporting a truthful location is a dominant strategy. Thus, we can safely assume that such mechanisms operate on the agents' true locations.

\paragraph{Approximation Ratio.}
Let $\phi: \R^n \to \R$ be a norm.
The $\phi$-norm \emph{social cost} of a facility $f$ aggregates the individual costs using the norm $\phi$:
\begin{equation*}
 \cost_\phi(f, x) = \phi \left (\|f - x_1\|_q, \|f - x_2\|_q, \dots, \|f - x_n\|_q \right ).
\end{equation*}
Given a set of locations $x \in (\R^d)^n$, we define the \textit{approximation ratio} of a mechanism $\M$ to be the worst-case ratio of the mechanism's social cost compared to the optimal social cost:
\begin{equation*}
    \alpha_\phi(\M) = \sup_{x \in (\R^{d})^n} \frac{\cost_\phi(\M(x), x)}{\min_{f \in \R^d} \cost_\phi(f, x)}.
\end{equation*} 

In this work, the norm $\phi$ will be \emph{monotone} (i.e., $\phi(u)\le \phi(v)$ if for all $j\in[n]$, $|u_j|\le |v_j|$) and \emph{symmetric} (i.e., $\phi$ is invariant under coordinate permutations).
Note that monotonicity implies invariance under coordinate sign
changes.
When $\phi$ is taken to be an $\ell_p$ norm, we will let $\cost_p$ and $\alpha_p$ denote the social cost and approximation ratio respectively.

\subsection{Coordinate-wise Median}
In this work, we focus on the coordinate-wise median (CM) mechanism. For a given location profile $x$, this mechanism calculates the standard one-dimensional median for each of the $d$ coordinates independently.
\footnote{When $n$ is even, the mechanism can break ties arbitrarily (in favor of either the larger or the smaller value), but consistently across dimensions. In this work, we assume that ties are broken in favor of the smaller value.}
If $x_{i,j}$ denotes the $j$-th coordinate of agent $i$'s location, the mechanism outputs a facility location $m = (m_1,m_2, \dots, m_d) \in \R^d$, where
\begin{equation*}
m_j = \operatorname{median}(x_{1,j}, x_{2,j}, \dots, x_{n,j}) \quad \text{for each } j \in [d].
\end{equation*}
It is well-known that the CM mechanism is strategyproof and \emph{anonymous} (i.e., for any permutation $\pi$ of the agents $[n]$, it holds that $\M(x_1,x_2, \dots, x_n) = \M(x_{\pi(1)},x_{\pi(2)}, \dots, x_{\pi(n)})$).

\paragraph{Simplifying Transformations.} 
If the CM mechanism receives a location profile $(z_1,z_2, \dots, z_n) \in (\R^d)^n$, we can make a sequence of transformations that do not affect our analysis of the approximation ratio. Let $f \in \R^d$ denote the optimal facility location under the $\phi$-norm social cost objective. We perform the following three steps for all norms in this work.

\begin{enumerate}
    \item \textit{Translation:} We uniformly translate all points so that $\M(z) = (0,0,\dots, 0)$. Since $\ell_q$-norms are translation invariant, the approximation ratio is preserved.

    \item \textit{Orientation (Non-negative Orthant):} Following the translation, we choose the orientation of the coordinate axes so that the optimal facility $f$ lies in the non-negative orthant. This does not affect individual distances, and hence social costs, since the $\ell_q$ distance is invariant under coordinate sign changes.
    
    \item \textit{Normalization:} If $f = \M(z) = 0$, then the mechanism is optimal on this instance. So, we can assume that $f \neq \M(z)$. We can normalize so that $\qnorm{f} = 1$. This does not affect the approximation ratio since all distances are scaled by the same factor.
\end{enumerate}
Additionally, for the $p$-norm social costs, we also perform the following transformation.
\begin{enumerate}[start=4]
    \item \textit{Duplication:} If $n$ is odd, we can consider the instance where each location appears twice (i.e., the profile $y = (z_1, z_1, z_2, z_2, \dots, z_n, z_n)$). The optimal facility location for this new instance is still $f$ and we also have $\M(y) = \M(z)$, so the approximation ratio is unchanged. Thus, we assume without loss of generality that $n$ is even.
\end{enumerate}

%% file: monotone-norm.tex
In this section, we bound the approximation ratio of the CM mechanism when the social cost is obtained by applying an arbitrary monotone symmetric norm to the vector of individual distances. 
This class of norms spans a large set of objectives that one could use to evaluate the effectiveness of the coordinate-wise median, including the typical utilitarian and egalitarian social costs.

Recently, \citet{Gravin:2025aa} showed that the CM mechanism is at most a $3$-approximation for the utilitarian social cost, a significant improvement over the $O(\sqrt{d})$ upper bound of \citet{Meir:2019aa} which was widely believed to be tight. Our main result in this section shows that this phenomenon is not limited to the utilitarian social cost: for any monotone symmetric norm, the CM mechanism is at most a $3$-approximation.

\thmmonotonesocialcost*

To prove this theorem, we use weak majorization and the Ky Fan dominance theorem
(see, e.g., \citealp{Hardy:1952aa,marshall_inequalities_2011}). For vectors $u,v\in\R_{\ge0}^n$, write $u\preceq v$ (and say $v$ \textit{weakly majorizes} $u$) if
\[
    \sum_{i=1}^k u_i^\downarrow \le \sum_{i=1}^k v_i^\downarrow
    \qquad\text{for every } k\in[n],
\]
where $u^\downarrow$ and $v^\downarrow$ denote the nonincreasing
rearrangements. The standard Ky Fan dominance theorem states that if $u\preceq v$ and $\phi$ is a monotone symmetric
norm, then $\phi(u)\le \phi(v)$. 
Since norms are homogeneous, the following lemma is immediate.

\begin{lemma}\label{lem:majorization-reduction}
    For any monotone symmetric norm $\phi$ on $\R^n$, if $u,v\in\R_{\ge0}^n$ and
$u\preceq Cv$ for some $C\ge0$, then $\phi(u)\le C\phi(v)$.
\end{lemma}

Our approach is to find appropriate values of the constant $C$ in the above lemma that suffice for our setting.
In the following proposition, we show that $C=3$ suffices in general.
\begin{proposition}\label{prop:monotone-majorization-high-d}
Suppose that $\M(x)=0$ and $f\ge0$. Then
\[
    \left(\qnorm{x_1},\qnorm{x_2},\dots,\qnorm{x_n}\right)
    \preceq
    3\left(\qnorm{x_1-f},\qnorm{x_2-f},\dots,\qnorm{x_n-f}\right).
\]
\end{proposition}

\begin{proof}
Writing $u=(\qnorm{x_1},\qnorm{x_2},\dots,\qnorm{x_n})$ and
    $v=(\qnorm{x_1-f},\qnorm{x_2-f},\dots,\qnorm{x_n-f})$, we aim to show that $u\preceq 3v$.
Let $q^*$ satisfy
$1/q+1/q^*=1$ (with $1^*=\infty$ and $\infty^*=1$). By the
dual characterization of the $\ell_q$ norm, choose $w\ge0$ with
$\qnorm[q^*]{w}=1$ and $w\cdot f=\qnorm{f}$. 

For each agent $i$, H\"older's
inequality gives $w\cdot |x_i-f|\le \qnorm[q^*]{w}\qnorm{x_i-f}=v_i$.
Since $w\ge0$ and $f\ge0$, it follows that
\[
    v_i
    \ge w\cdot |x_i-f|
    \ge \sum_{j:x_{i,j}\le0} w_j f_j .
\]
Summing over agents, and using the fact that every coordinate has at least
$n/2$ nonpositive entries when the coordinate-wise median is zero, gives
\[
    \sum_{i=1}^n v_i
    \ge \frac n2\sum_{j=1}^d w_j f_j
    = \frac n2 \qnorm{f}.
\]
Consequently, for every $k\in[n]$,
\begin{equation}\label{eq:lb-sum-v}
    \sum_{i=1}^k v_i^\downarrow
    \ge \frac{k}{n}\sum_{i=1}^n v_i
    \ge \frac{k}{2}\qnorm{f}.
\end{equation}
Now let $T$ be the set of the $k$ largest entries among $u_1,u_2,\dots,u_n$.
We readily derive that 

\begin{align*}
    \sum_{i=1}^k u_i^\downarrow
    &= \sum_{i\in T} \qnorm{x_i} \\
    &\le \sum_{i\in T} \qnorm{x_i-f} + k\qnorm{f}
        \tag{triangle inequality} \\
    &\le \sum_{i=1}^k v_i^\downarrow + 2\sum_{i=1}^k v_i^\downarrow
        \tag{by \eqref{eq:lb-sum-v}} \\
    &= 3\sum_{i=1}^k v_i^\downarrow .
\end{align*}
Since this holds for every $k\in[n]$, the desired weak-majorization inequality
follows.
\end{proof}

Already, we have a factor 3 approximation for the CM mechanism for any monotone symmetric norm.
However, upon closer inspection, there is additional structure that we can exploit
when $d = 2$. It turns out that the CM mechanism enforces a balance constraint across opposite orthants.

\paragraph{Balance Constraints in $\R^2$}
Suppose that
$x\in(\R^2)^n$ satisfies $\M(x)=0$. Each location $x_i\in\R^2$ is assigned a
signature $\sigma(x_i)\in\{-1,1\}^2$, with zero coordinates counted according
to the median constraint. Let
\begin{align*}
    &S_{++} = \{ x_i \mid  \sigma(x_i) = (+1, +1) \}, \qquad
    &S_{-+} = \{x_i \mid  \sigma(x_i) = (-1, +1) \}, \\
    &S_{--} = \{x_i \mid  \sigma(x_i) = (-1, -1) \}, \qquad
    &S_{+-} = \{x_i \mid  \sigma(x_i) = (+1, -1) \}.
\end{align*}
Let $t=\lfloor n/2\rfloor$. The median constraint can be written as
\begin{subequations}\label{eq:planar-balance-constraints}
\begin{align}
    |S_{++}| + |S_{+-}| &= t, \label{eq:planar-balance-x}\\
    |S_{++}| + |S_{-+}| &= t, \label{eq:planar-balance-y}\\
    |S_{++}| + |S_{+-}| + |S_{-+}| + |S_{--}| &= n .
    \label{eq:planar-balance-total}
\end{align}
\end{subequations}
By subtracting \eqref{eq:planar-balance-y} from
\eqref{eq:planar-balance-x}, we get $|S_{+-}|=|S_{-+}|$. Substituting this
identity into \eqref{eq:planar-balance-total} yields
$|S_{--}|=|S_{++}|$ if $n$ is even and
$|S_{--}|=|S_{++}|+1$ if $n$ is odd. Thus all points can be paired with a point
in the opposite signed orthant, except possibly for one leftover point in
$S_{--}$. With this balance constraint, we can prove the following more refined
majorization inequality for $d=2$.

\begin{proposition}\label{prop:monotone-majorization-2d}
Suppose that $d=2$, $\M(x)=0$, and $f\ge0$. Then
\[
    \left(\qnorm{x_1},\qnorm{x_2},\dots,\qnorm{x_n}\right)
    \preceq
    2\left(\qnorm{x_1-f},\qnorm{x_2-f},\dots,\qnorm{x_n-f}\right).
\]
\end{proposition}

\begin{proof}
By the balance constraint, all points can be paired with a point in the
opposite signed orthant, except possibly for one leftover point in $S_{--}$.
Consider any opposite pair $(a,b)$. Since $a$ and $b$ have opposite signs in
every coordinate, $|a|,|b|\le |a-b|$ coordinate-wise. By monotonicity of the
$q$-norm and the triangle inequality,
\[
    \qnorm{a},\qnorm{b}\le \qnorm{a-b}
    \le \qnorm{a-f}+\qnorm{b-f}.
\]
Therefore, if $u=(\qnorm{a},\qnorm{b})$ and
$v=(\qnorm{a-f},\qnorm{b-f})$, then $u\preceq 2v$. Indeed, the one-coordinate
inequality follows from
\[
    \max\{\qnorm{a},\qnorm{b}\}
    \le \qnorm{a-f}+\qnorm{b-f}
    \le 2\max\{\qnorm{a-f},\qnorm{b-f}\},
\]
while the two-coordinate inequality follows from
\[
    \qnorm{a}+\qnorm{b}
    \le 2(\qnorm{a-f}+\qnorm{b-f}).
\]

If there is a leftover point $z\in S_{--}$, then $z\le0\le f$, so
$|z|\le |z-f|$ coordinate-wise and hence $\qnorm{z}\le\qnorm{z-f}$. Combining
the pairwise weak-majorization inequalities over all pairs, and including the
possible leftover point, gives the desired weak-majorization inequality.
\end{proof}

Putting it all together, our main result follows.
\begin{proof}[Proof of \Cref{thm:monotone-social-cost}]
By the reductions in the preliminaries, we may assume $\M(x)=0$ and $f\ge0$.
\Cref{prop:monotone-majorization-high-d} and
\Cref{lem:majorization-reduction} imply
\[
    \cost_\phi(\M(x),x)\le 3\cost_\phi(f,x)
\]
for every comparison facility $f$, so $\alpha_\phi(\M)\le3$.
When $d=2$, \Cref{prop:monotone-majorization-2d} and
\Cref{lem:majorization-reduction} similarly imply
$\cost_\phi(\M(x),x)\le2\cost_\phi(f,x)$, so $\alpha_\phi(\M)\le2$.
\end{proof}

In the remainder of the paper, we specialize to $p$-norm social costs. This
additional structure allows us to obtain sharper guarantees as a function of
$p$ and $q$.

%% file: program.tex
Our analysis of the CM mechanism for $p$-norm social costs will be based on a generalization of a mathematical program introduced by \citet{Gravin:2025aa}. In this section, we develop this program and explain how it allows us to derive upper bounds on the approximation ratio. In what follows, we let $\M(z)$ denote the output of the mechanism.

Fix a tuple of reported locations $(z_1, \dots, z_n) \in (\R^d)^n$ and let $f \in \R^d$ denote the optimal facility location under the $p$-norm social cost objective. 
We make the standard transformations to simplify our analysis (translation, orientation, normalization, and duplication).

For a fixed facility location $f \neq 0$ in the nonnegative orthant and any parameter $\lambda$, we can ask what is the smallest possible value of $\sum_{i\in [n]} \qnorm{x_i - f}^p - \lambda^p \qnorm{x_i}^p$
among any set of points $x = (x_1,\dots, x_n) \in (\R^d)^n$ satisfying $\M(x) = 0$. If this value is nonnegative, then $\cost_p(\M(z), z) \leq (1 /\lambda) \cdot\cost_p(f, z)$ and thus $\M$ is a $(1/\lambda)$-approximation on the original instance.

Our goal therefore is to fix a facility location and find the largest value $\lambda \in (0, 1)$ such that the following program has a nonnegative minimum value.
\begin{subequations} \label{eq:full-program}
\begin{align}
 \min_{x \in  (\R^d)^n} \quad & \sum_{i=1}^n g(x_i), \qquad \text{where }  g(v) := \qnorm{v - f}^p - \lambda^p \qnorm{v}^p  \label{eq:objective}  \\
 \text{s.t.} \quad & \qnorm{f} = 1, \quad f \geq 0  \label{eq:domain-constraint}
\end{align}

To encode the constraint that the origin is the output of the CM mechanism, we define a \emph{signature} mapping $\sigma(v) = (\operatorname{sign}(v_1), \dots, \operatorname{sign}(v_d)) \in \{-1, 1\}^d$. For $a \neq 0$, $\operatorname{sign}(a)$ is the standard sign function. Entries equal to zero may receive either sign, according to how they are counted by the CM mechanism when selecting the median; we denote these two cases by $0^-$ and $0^+$. Thus, whenever the CM output is the origin and $n$ is even, the zero entries can be assigned signs so that each coordinate has an equal balance of positive and negative signs.
We encode this median constraint as
\begin{equation}
    \sum_{i=1}^n \sigma(x_i) = 0.  \label{eq:median-constraint}
\end{equation}
\end{subequations}

%% file: plane.tex
We start by restating the initial program to handle the case when $d = 2$.
Recall that when the CM mechanism is given an input $x \in (\R^2)^n$, each location $x_i \in \R^2$ is assigned a signature $\sigma(x_i) \in \{-1, 1\}^2$. The median constraint ensures that we can partition the locations into the following sets based on these signatures:
\begin{align*}
    &S_{++} = \{ x_i \mid  \sigma(x_i) = (+1, +1) \}, \qquad
    &S_{-+} = \{x_i \mid  \sigma(x_i) = (-1, +1) \}, \\
    &S_{--} = \{x_i \mid  \sigma(x_i) = (-1, -1) \}, \qquad
    &S_{+-} = \{x_i \mid  \sigma(x_i) = (+1, -1) \}.
\end{align*}
Using these defined sets, our objective \eqref{eq:objective} becomes
\begin{subequations}\label{eq:full-program-2d}
\begin{align}\label{eq:objective-2d}
\min_{x \in  (\R^2)^n} \quad &
\sum_{x_i \in S_{++}} g(x_i) + \sum_{x_i \in S_{--}} g(x_i) +
\sum_{x_i \in S_{+-}} g(x_i) + \sum_{x_i \in S_{-+}} g(x_i).
\end{align}

By the balance constraint established in \Cref{sec:monotone-norm}, and using
the duplication reduction so that $n$ is even, the median constraints reduce to
\begin{align}
    |S_{++}| &= |S_{--}|, \label{eq:median-constraint-2d-1}\\
    |S_{-+}| &= |S_{+-}|, \label{eq:median-constraint-2d-2}\\
    |S_{++}| + |S_{+-}| &= \frac{n}{2}. \label{eq:median-constraint-2d-3}
\end{align}
This formulation makes two key structural features immediately apparent.
First, there is a balance requirement between opposite orthants. Second, the
system has only a single degree of freedom, since fixing the size of one
partition (for instance, $|S_{++}|$) automatically determines the remaining
three. Both of these features will be central to our subsequent analysis.

To finish specifying our program, we maintain constraint
\eqref{eq:domain-constraint}:
\begin{align}
    \qnorm{f} = 1, \quad f \ge 0 .
\end{align}
\end{subequations}
As mentioned in \Cref{sec:program}, our goal is to find a large value $\lambda \in (0, 1)$ that guarantees a nonnegative minimum for our objective. Note that for a fixed facility location $f \ge 0$, evaluating the contribution $g(x_i)$ of a single point in isolation could yield a large negative value, which would drastically weaken our lower bound for $\lambda$. We circumvent this issue by leveraging the symmetry established by the median constraints. By pairing points from opposite orthants, we can offset extreme individual costs and instead lower bound their joint contribution to the objective.

\subsection{Optimization over Opposing Orthants}
We begin by analyzing the points in the opposing orthants. Because of the balance constraints established in \eqref{eq:median-constraint-2d-1} and \eqref{eq:median-constraint-2d-2}, we can group the terms of our objective \eqref{eq:objective-2d} into paired components:
\begin{align*}
   \left ( \sum_{x_i \in S_{++}} g(x_i) + \sum_{x_i \in S_{--}} g(x_i)  \right)
   + \left (
    \sum_{x_i \in S_{+-}} g(x_i) + \sum_{x_i \in S_{-+}} g(x_i)
    \right ).
\end{align*}
For each component, the equal sizes of the opposite sets allow us to form a convenient bijection. By pairing each point $a \in S_{++}$ with another point $b \in S_{--}$, and similarly for $S_{+-}$ and $S_{-+}$, we effectively rewrite the entire objective as a sum of pairwise contributions. As a result, determining when \eqref{eq:objective-2d} is nonnegative reduces to lower bounding the joint cost $g(a) + g(b)$ of any such opposite pair $(a, b)$. We establish this bound using the following lemma.

\begin{restatable}{lemma}{lemminsumopp}\label{lem:min-sum-opp}
Let $f \geq 0$ be a facility location. For any pair of points $(a, b)$ from opposite orthants, i.e., $(a, b) \in (S_{+-} \times S_{-+}) \cup (S_{++} \times S_{--})$, the following statements hold:
\begin{enumerate}
\item If $1 \leq q \leq p$, then $\qnorm{a - f}^p + \qnorm{b - f}^p \geq 2^{1-p} \left ( \qnorm{a}^p + \qnorm{b}^p \right )$.
\item If $1 \leq p < q$, then $\qnorm{a - f}^p + \qnorm{b - f}^p \geq 2^{p/q - p} \left ( \qnorm{a}^p + \qnorm{b}^p \right )$.
\end{enumerate}
\end{restatable}
\begin{proof}
We first note that since $a$ and $b$ are in opposing orthants, they have opposite signs in each coordinate (or are zero). Thus, for any coordinate $i \in \{1, 2\}$, we have $|a_i - b_i| = |a_i| + |b_i|$.
This guarantees that for any $r \geq 1$, the $r$-norm of their difference simplifies to:
\[ \qnorm[r]{a - b}^r = (|a_1| + |b_1|)^r + (|a_2| + |b_2|)^r .\]
We will use this observation when proving each of the statements.

We begin with the first statement.
Suppose that $1 \leq q \leq p$ and let $h(z)=\qnorm{z}^p$. Since $h$ is
convex, we can derive that
\begin{align*}
\qnorm{a - f}^p + \qnorm{b - f}^p
&= h(a-f)+h(f-b) \\
&\geq 2h\left(\frac{(a-f)+(f-b)}{2}\right) \tag{convexity of $h$} \\
&= 2^{1-p}\qnorm{a-b}^p \\
&= 2^{1-p} \left( \qnorm{a - b}^q \right)^{p/q} \\
&= 2^{1-p} \left( (|a_1| + |b_1|)^q + (|a_2| + |b_2|)^q \right)^{p/q} \\
&\geq 2^{1-p} \left( |a_1|^q + |b_1|^q + |a_2|^q + |b_2|^q \right)^{p/q} \\
&= 2^{1-p} \left( \qnorm{a}^q + \qnorm{b}^q \right)^{p/q} \\
&\geq 2^{1-p} \left( \qnorm{a}^p + \qnorm{b}^p \right).
\end{align*}

Now suppose that $1 \leq p < q$. We can again derive that
\begin{align*}
\qnorm{a - f}^p + \qnorm{b - f}^p
&\geq 2^{p/q - 1} \left( \qnorm[p]{a - f}^p + \qnorm[p]{b - f}^p \right) \tag{norm comparison in $\R^2$} \\
&\geq 2^{p/q - 1} \cdot 2^{1-p} \left( \qnorm[p]{a - f} + \qnorm[p]{b - f} \right)^p \tag{Jensen's inequality} \\
&\geq 2^{p/q - p} \qnorm[p]{(a - f) + (f - b)}^p \tag{Triangle inequality} \\
&= 2^{p/q - p} \qnorm[p]{a - b}^p \\
&= 2^{p/q - p} \left( (|a_1|+|b_1|)^p + (|a_2|+|b_2|)^p \right) \\
&\geq 2^{p/q - p} \left( |a_1|^p + |b_1|^p + |a_2|^p + |b_2|^p \right) \\
&= 2^{p/q - p} \left( \qnorm[p]{a}^p + \qnorm[p]{b}^p \right) \\
&\geq 2^{p/q - p} \left( \qnorm{a}^p + \qnorm{b}^p \right) \tag{since $\qnorm{z} \leq \qnorm[p]{z}$ for $z \in \R^2$}
\end{align*}
Both statements hold as desired.
\end{proof}

Recall the definition of our objective term, $g(z) = \qnorm{z - f}^p - \lambda^p \qnorm{z}^p$. Applying \Cref{lem:min-sum-opp} to any opposite pair $(a, b)$ directly implies that their combined contribution, $g(a) + g(b)$, is nonnegative for the specified values of $\lambda$.
Because the entire objective \eqref{eq:objective-2d} can be expressed as a sum of these paired terms, we immediately arrive at the following corollary.

\begin{corollary}\label{cor:opposite-orthants}
The objective function \eqref{eq:objective-2d} is nonnegative when
\begin{equation*}
\lambda =
\begin{cases}
2^{1/p - 1} & \text{if } 1 \leq q \leq p,  \\
2^{1/q - 1} & \text{if } 1 \leq p < q.
\end{cases}
\end{equation*}
\end{corollary}

\subsection{Approximation Ratio}

We are now ready to derive the exact approximation ratio of the CM mechanism.
\thmupperboundtwo*
\begin{proof}
We write the simplified program given at the start of \Cref{sec:plane}.
We then recall that when $\lambda$ is as specified in \Cref{cor:opposite-orthants}, the objective value \eqref{eq:objective-2d} is nonnegative for any $x \in (\mathbb R^2)^n$ that satisfies the program constraints.
For such $\lambda$ and any $x \in (\mathbb R^2)^n$, we have $\cost_p(f, x)^p \geq \lambda^p \cdot \cost_p(\M(x), x)^p$, which implies our desired result.
\end{proof}
A notable consequence of \Cref{thm:cm-2d-upper-bound} is that it determines the asymptotic approximation ratio of the CM mechanism for $p$-norm social costs in $\ell_2(\R^2)$, resolving a conjecture of \citet{Goel:2023aa}. We also provide approximation guarantees for all two-dimensional normed spaces $\ell_q(\R^2)$. Qualitatively, our result reveals a phase transition determined by the relationship between $p$ and $q$, highlighting a tension between the $p$-norm social cost and $q$-norm distances.

Below, we establish asymptotically matching lower bounds in all parameter regimes.

\begin{theorem}[Lower Bound]\label{thm:cm-2d-lower-bound}
In the normed space $\ell_q(\R^2)$ with $p,q \geq 1$, the CM mechanism $\M$ satisfies the following lower bounds:
\begin{align*}
\sup_{n \in \mathbb N}  \alpha_p(\M) \geq  \begin{cases} 2^{1 - 1/p}  & \text{if } 1 \leq q \leq p, \\ 2^{1 - 1/q} & \text{if } 1 \leq p < q. \end{cases}
\end{align*}
\end{theorem}

\begin{proof}
Our analysis proceeds in two cases.

\paragraph{Case 1: $1 \leq q \leq p$.}
Extending the analysis of \citet{Feigenbaum:2017aa, Goel:2023aa}, we in fact show that any deterministic, strategyproof mechanism $\A$ has $\sup_n \alpha_p(\A) \geq 2^{1 - 1/p}$ in this regime.
Consider a profile $x \in (\R^2)^2$ with $x_1 = (-1, 0)$ and $x_2 = (1, 0)$.
It is not hard to see that the optimal facility location is the point $f = (0, 0)$
and that $\cost_p(f, x) = 2^{1/p}$.
If $\A(x) = (1, 0)$ then we have
\begin{align*}
    \frac{\cost_p(\A(x), x)}{\cost_p(f, x)} = \frac{2}{2^{1/p}} = 2^{1-1/p}.
\end{align*}
If $\A(x) \neq (1, 0)$, we can consider the profile $x' \in (\R^2)^2$ with $x_1' = \A(x)$ and $x'_2 = (1, 0)$.
It is not hard to see that the optimal facility location for the profile $x'$ is the point $f' = (x_1' + x_2')/2$, which also satisfies the following:
\[\qnorm{f' - x_1'} = \qnorm{f' - x_2'} = \qnorm{x_1' - x_2'}/ 2.\]
Since $\A$ is strategyproof, we know that
$0 = \qnorm{\A(x) - x_1'} \geq \qnorm{\A(x') - x_1'}$. This in turn implies that $\A(x') = x_1'$. Thus, we can derive that
\begin{align*}
\frac{\cost_p(\A(x'), x')}{ \cost_p(f', x')}
&=\frac{ \left (\qnorm{\A(x') - x_1'}^p + \qnorm{\A(x') - x_2'}^p \right )^{1/p}}{\left (\qnorm{f'- x_1'}^p + \qnorm{f' - x_2'}^p \right )^{1/p}} \\
&= \frac{\qnorm{x_1' - x_2'}}{ 2^{(1-p)/p}\qnorm{x_1' - x_2'}} \\
&= 2^{1 - 1/p}.
\end{align*}
This bound applies independently of $q$; however, it is the dominant lower bound when $1 \leq q \leq p$.

\paragraph{Case 2: $1 \leq p < q$.} For each integer $k > 0$, consider the profile $x \in (\R^2)^{2k + 1}$ with $k$ agents located at the point $(1, 0)$, $k$ agents at the point $(0, 1)$ and $1$ agent at the point $(0, 0)$.
It is easily verified that $\M(x) = (0, 0)$, as the median of $k$ ones, $k$ zeros, and a single zero is zero. Consider the facility location $f' = (1/2, 1/2)$. The optimal facility location $f$ is such that
\begin{align*}
\frac{\cost_p(\M(x), x)}{\cost_p(f, x)}
\geq \frac{(2k)^{1/p}}{\cost_p(f', x)}
= \frac{(2k)^{1/p}}{2^{1/q - 1}  (2k + 1)^{1/p} }.
\end{align*}
As $k \to \infty$, the rightmost ratio tends to $2^{1 - 1/q}$, which implies our desired result.
\end{proof}

Finally, we comment on the optimality of the CM mechanism. \citet{Goel:2023aa} proved its optimality among deterministic, strategyproof, and anonymous mechanisms in the Euclidean space $\ell_2(\R^2)$ for any $p \geq 1$ (see Theorem 1 in their work). In \Cref{thm:cm-2d-lower-bound}, we show that the CM mechanism continues to be optimal in all normed spaces $\ell_q(\R^2)$, whenever $1 \leq q \leq p$. Whether this optimality extends to the $1\leq p <q$ regime presents an interesting direction for future research.

%% file: higher-dimension.tex
We now turn our attention to the facility location setting in higher dimensions $(d \geq 3)$.
Recently, \citet{Gravin:2025aa} established a constant approximation for the facility location problem in $\ell_q(\R^d)$ under utilitarian social cost ($p = 1$). 

\begin{theorem}[\citealp{Gravin:2025aa}]\label{thm:1-norm}
    In the normed space $\ell_q(\R^d)$ with $q \geq 1$ and $d \geq 1$, the CM mechanism $\M$ has approximation ratio $\alpha_1(\M) \leq \mathtt{UB}(q)$, where $\mathtt{UB}(q)$ is an increasing function with $\mathtt{UB}(1) = 1$ and $\lim_{q \to \infty} \mathtt{UB}(q) = 3$.
\end{theorem}

From our earlier analysis for monotone symmetric norms, the approximation ratio for the CM under the $p$-norm social cost will never exceed $3$. 
In this section, we use a more careful analysis of our initial program \eqref{eq:full-program} to give more refined approximation guarantees for this $p$-norm objective.

\thmupperboundhigh*

Our analysis for the $d \ge 3$ case will be markedly different from the $d = 2$ case discussed in the previous section. This is because our new program no longer admits a clean restatement of the median constraints. In general, the median mechanism for $\mathbb{R}^d$ partitions agents across $2^d$ possible signatures, but only enforces $d + 1$ constraints (namely, one for each dimension and one for the total number of agents). When $d = 2$, we could exploit a single degree of freedom and argue based on bijections between opposite orthants. This argument breaks down in higher dimensions.

Because our geometric arguments for $d=2$ do not generalize, our approach here is to analyze an appropriate relaxation of our program.
We adopt the proof structure introduced by \citet{Gravin:2025aa}; however, our technical analysis diverges in several places to handle the tension between the $p$-norm for social cost and the $q$-norm for distances.

\subsection{Upper Bounds on Approximation Ratio}
We first note that our program is minimized at some point in the domain.
\begin{claim}
The program \eqref{eq:full-program} attains its minimum at some $x \in (\R^d)^n$.
\end{claim}
\begin{proof}
By the reverse triangle inequality and recalling $\qnorm{f} = 1$, whenever $\qnorm{x_i}>1$ we have
\[
g(x_i) = \qnorm{x_i - f}^p - \lambda^p \qnorm{x_i}^p \ge (\qnorm{x_i} - 1)^p - \lambda^p \qnorm{x_i}^p.
\]
We further factor out $\qnorm{x_i}^p$ to obtain
\[
g(x_i) \ge \qnorm{x_i}^p \left( \left|1 - \qnorm{x_i}^{-1}\right|^p - \lambda^p \right).
\]

Since $\lambda \in (0, 1)$, the bracketed expression approaches a strictly positive constant $(1 - \lambda^p)$ as $\qnorm{x_i} \to \infty$, implying $g(x_i) \to \infty$.
Because $g$ is continuous and diverges to infinity, it attains a global minimum $g_{\min} \in \mathbb{R}$. We can therefore choose a radius $R$ large enough such that $g(x_i) > n - (n-1)g_{\min}$ for any $x_i \notin B_R(0)$. If a location profile $x$ has even one point outside this closed ball, its total objective value strictly exceeds the trivial solution of placing all agents at the origin:
\[
\sum_{i=1}^n g(x_i) > \left(n - (n-1)g_{\min}\right) + (n-1)g_{\min} = n = \sum_{i=1}^n g(0).
\]
Thus, we can restrict our minimization to the compact set $\left\{x \in B_R(0)^n \mid \sum_{i=1}^n \sigma(x_i) = 0\right\}$, guaranteeing the minimum is attained.
\end{proof}
Following \citep{Gravin:2025aa}, we analyze the local properties of this optimal solution $x$. By optimality, each component $x_i$ must minimize $g$ over all points sharing its signature $\sigma(x_i)$. Therefore, our next step is to establish a pointwise lower bound on $g$ for any fixed signature.

\subsubsection{Optimization over Individual Points}
Fix an arbitrary signature, and let $z \in \mathbb{R}^d$ be a minimizer of $g$ restricted to this signature.
For our analysis, we assume $q > 1$.\footnote{The strict inequality $q > 1$ ensures that the map $z \mapsto |z|^q$ is everywhere differentiable, and prevents terms with $q - 1$ in the denominator, such as in the definition of $c(z)$ below, from becoming undefined. The case of $q =1$ follows from a limiting argument.}
We first observe that any coordinate with a positive sign naturally falls outside the interval $[0^+, f_j)$.
\begin{claim}\label{clm:wlog-points}
   We can assume without loss of generality that $z_j \in [f_j, +\infty)$ whenever $z_{j} \geq 0^+$.
\end{claim}
\begin{proof}
Recall that the function $g$ on input $z$ is defined as follows:
\[
g(z) = \qnorm{z - f}^p - \lambda^p \qnorm{z}^p.
\]
Suppose there exists a coordinate $j \in [d]$ such that $z_j \in [0^+, f_j)$. We analyze the local properties of $g(z)$ where $g$ is as defined in \eqref{eq:full-program}. Shifting $z_j$ to $f_j$ strictly decreases $|z_j - f_j|$, which in turn decreases $\qnorm{z - f}^p$. At the same time, this shift increases $|z_j|$, which in turn increases the penalty term $\qnorm{z}^p$. Together, these adjustments decrease $g(z)$ without altering the location's positive sign $\sigma(z)_j = +1$. Hence, we can assume the optimal location satisfies $z_j \ge f_j$.
\end{proof}

By introducing an intermediate quantity $c(z)$ defined in \eqref{eq:c-z}, we can further characterize the coordinates of $z$ in \Cref{lem:higher-d-optimal-z}.

\begin{equation}\label{eq:c-z}
    c(z) = \begin{cases} \lambda^{\frac{p}{q-1}} \cdot \left ( \frac{ \qnorm{z} }{ \qnorm{z - f} } \right )^{\frac{p-q}{q-1}} & \text{if} \ \qnorm{z - f} > 0, \\ 0 & \text{if} \ \qnorm{z - f} = 0. \end{cases}
\end{equation}

\begin{lemma}[Optimal $z$ with given signature]\label{lem:higher-d-optimal-z}
For every index $j \in [d]$ such that $z_j \neq 0$,
\begin{equation*}
    (1 - c(z))\cdot z_j = f_j
\end{equation*}
Furthermore, if $z_j > 0$ for some $j$, then $c(z) \in [0, 1)$. If $z_j < 0$ for some $j$, then $c(z) \geq 1$. Consequently, $z$ cannot simultaneously have strictly positive and strictly negative coordinates.
\end{lemma}
\begin{proof}
If $\qnorm{z - f} = 0$, then $z = f$ and $c(z) = 0$; the lemma is clearly satisfied in this case. For the remainder of the proof, assume $\qnorm{z - f} > 0$.
Note also that the claim clearly holds if $z = 0$, so we further assume $z \neq 0$.

Since $q > 1$, and since we have reduced to the case $z\neq 0$ and $z\neq f$, the function $g(z) = \qnorm{z - f}^p - \lambda^p \qnorm{z}^p$ is differentiable at $z$. Because $z$ is a local optimum of $g$, its partial derivatives with respect to any coordinate $j$ where $z_j \neq 0$ must equal zero. Computing the partial derivative, we obtain
\begin{equation}\label{eq:partial-derivative-g-app}
    \pdv{g}{z_j}(z) = p \qnorm{z - f}^{p-q} |z_j - f_j|^{q-1} \sign(z_j - f_j) - \lambda^p p \qnorm{z}^{p-q} |z_j|^{q-1} \sign(z_j).
\end{equation}

Now consider any $j$ such that $z_j \neq 0$. It must be the case that $\sign(z_j - f_j) \neq 0$ and thus $z_j \neq f_j$; otherwise the first term in \eqref{eq:partial-derivative-g-app} would vanish while the second would remain nonzero, contradicting stationarity.
Furthermore, since every term on the right hand side of \eqref{eq:partial-derivative-g-app} is nonzero, we must have $\sign(z_j - f_j) = \sign(z_j)$ to yield a derivative of zero. This sign agreement allows us to derive that
\begin{equation*}
    z_j - f_j = z_j \cdot \lambda^{\frac{p}{q - 1}} \left ( \frac{ \qnorm{z} }{ \qnorm{z - f} } \right )^{\frac{p-q}{q - 1}} = z_j \cdot c(z).
\end{equation*}

If $z_j > 0$, we have $c(z) = \frac{z_j - f_j}{z_j} = 1 - \frac{f_j}{z_j} \le 1$ (recall that \Cref{clm:wlog-points} implies $z_j \geq f_j$). Because $z_j \neq f_j$, this yields $c(z) < 1$.
On the other hand, if $z_j < 0$, we have $c(z) = \frac{z_j - f_j}{z_j} = 1 + \frac{f_j}{|z_j|} \ge 1$.
Because $c(z)$ is a single global scalar, it is impossible for $c(z) < 1$ and $c(z) \ge 1$ to hold simultaneously. Thus, $z$ cannot have both strictly positive and strictly negative coordinates. \qedhere
\end{proof}

\subsubsection{Program Relaxation}
Let $x = (x_1, \dots, x_n) \in (\mathbb R^d)^n$ be an optimal solution to the program specified in \eqref{eq:full-program}. We will use our characterization of the individual locations $x_i$ (given in the previous subsection) to write a relaxation of our initial program. To this end, for each $i \in [n]$, we define the variable
\begin{equation}
X_i = \sum_{\substack{j \in [d] \\ x_{ij} \geq 0^+}} f_j^q .
\end{equation}
Note that $X_i \in [0, 1]$ because $\qnorm{f}=1$.
We further introduce a function $G : [0, 1] \to \mathbb R$ that will help to establish a pointwise lower bound on $g(x_i)$. The function is given by
\begin{subequations}\label{eq:G-and-C-star-def}
\begin{equation}\label{eq:G-def}
G(Z) = \begin{cases}
    - Z^{p/q}  \lambda^{p} \left ( 1 - C^* \right )^{1-p} + (1 - Z)^{p/q}  & \text{if } p \geq q > 1 \\
    - Z^{p/q}  \lambda^{p} \left ( 1 - C^* \right )^{1-p} + (1 - Z)^{p/q} \left (1 - \gamma^{\frac{p}{1-p/q}} \right )^{1-p/q}   & \text{if } 1 \leq p < q
\end{cases}
\end{equation}
where $\gamma \in (\lambda, 1)$ is a parameter that will later be specified and $C^*$ is given by
\begin{equation}\label{eq:C-star-def}
C^* \coloneqq C^*(p, q, \gamma) = \begin{cases}
0 & \text{if } p = 1, \\
\lambda^\frac{p}{p-1} & \text{if } p \geq q > 1,  \\
\left ( \lambda / \gamma \right )^ \frac{p}{p-1} & \text{if } 1 < p < q.
\end{cases}
\end{equation}
\end{subequations}
The parameter $\gamma$ is introduced for technical reasons which will become apparent in later proofs.

Having defined $X_i$ and $G$, we can establish two key results. First, we show that the median mechanism induces a convenient structure on the sum of the variables $X_i$. Then, we show that $g(x_i)$ is bounded below by $G(X_i)$. The proof of the latter (which follows from an involved case analysis of $g$ and $G$ dependent on the relation between $p$ and $q$) is deferred to \Cref{app:higher-d-g-lowerbound}.
\begin{lemma}\label{lem:sum-X-i}
The variables $X_1, \dots, X_n$ are such that $\sum_{i \in [n]} X_i = n / 2$.
\end{lemma}

\begin{proof}
By summing over $X_i$ and rearranging summations, we can derive that
\begin{align*}
    \sum_{i \in [n]} X_i =
    \sum_{i \in [n]} \sum_{\substack{j \in [d] \\ x_{ij} \geq 0^+}} f_j^q =
    \sum_{j \in [d]} \sum_{\substack{i \in [n] \\ x_{ij} \geq 0^+}} f_j^q =
    \sum_{j \in [d]} f_j^q \sum_{\substack{i \in [n] \\ x_{ij} \geq 0^+}} 1 =
    \frac{n}{2}.
\end{align*}
The final equality follows from the fact that $\| f\|_q =1$ and the fact that exactly half of the locations are assigned a positive sign for each coordinate $j \in [d]$.
\end{proof}

\begin{restatable}{lemma}{lemglowerbound}\label{lem:g-lower-bound}
For every $i \in [n]$, we have $g(x_i) \geq G(X_i)$.
\end{restatable}
Granting \Cref{lem:sum-X-i} and \Cref{lem:g-lower-bound}, it is not hard to see that our initial program can be relaxed to the following program:
\begin{subequations}\label{eq:relaxed-program-higher-d}
\begin{align}
 \min_{(X_1, \dots, X_n) \in [0, 1]^n} \quad & \sum_{i \in [n]} G(X_i) \label{eq:relaxed-objective}  \\
 \text{s.t.} \quad & \sum_{i \in [n]} X_i = \frac{n}{2}  \label{eq:relaxed-median-constraint}.
\end{align}
\end{subequations}
In later sections, we will analyze the relaxed program based on the cases of $p$ and $q$.
Our goal is to determine when the solution of the relaxed program is nonnegative.
Since the relaxed program provides a lower bound for our initial program, we can then find a sufficient condition on $\lambda \in (0, 1)$ that makes the initial program nonnegative.

\newcommand{\h}{\textcolor{black}{\gamma^{\frac{p}{1-r}}}}

\clearpage

\subsubsection{The Case of \texorpdfstring{$p/q \in \{1, 2\}$}{p/q = 1, 2}}
We first consider the setting in which $p/q \in \{1, 2\}$. It turns out that analysis of the relaxed program \eqref{eq:relaxed-program-higher-d} is relatively straightforward. By standard convexity arguments, we can establish a range of $\lambda$ values for which the relaxed objective is nonnegative, and hence the initial program \eqref{eq:full-program} is nonnegative. Taking the largest such value of $\lambda$ gives an upper bound on the approximation ratio of the CM mechanism.

\begin{proposition}\label{prop:p-over-q-is-1-or-2}
If $p/q \in \{1, 2 \}$, then $\mathtt{UB}(p, q) = 2^{1 - 1/p}$.
\end{proposition}
\begin{proof}
Let $r = p/q$. It suffices to show that for all $\lambda \in (0, 2^{1/p - 1}]$, the minimum of \eqref{eq:relaxed-program-higher-d} is nonnegative. This would mean that we can take $\mathtt{UB}(p,q) = 1 / \lambda^* \coloneqq 1 / (2^{1/p - 1}) = 2^{1 - 1/p}$.

Since $r \in \{1, 2\}$, we are in the setting where $p \ge q > 1$. From \eqref{eq:G-and-C-star-def}, we can rewrite $G(Z)$ as
\begin{equation*}
G(Z) = \delta_1 ( 1 - Z)^r - \delta_2 Z^r,
\end{equation*}
where $\delta_1 = 1$ and $\delta_2 = \lambda^p (1-C^*)^{1-p}$.
Below, we will repeatedly use the fact that $X_i \in [0,1]$ for every $i \in [n]$ and the fact that the median constraint \eqref{eq:relaxed-median-constraint} implies $\sum_{i \in [n]} X_i = \frac{n}{2}$.

If $r = 1$, then
\begin{align*}
\sum_{i \in [n]} G(X_i)
= \sum_{i \in [n]} \left ( \delta_1 (1 -X_i) -\delta_2 X_i \right )
= \frac{n}{2} \left ( \delta_1 - \delta_2 \right ).
\end{align*}
If $r = 2$, then $G$ can either be concave, convex or linear. If $G$ is concave, then
\begin{align*}
\sum_{i \in [n]} G(X_i)
\geq \sum_{i \in [n]} \left ( X_i \cdot G(1) + (1- X_i) \cdot G(0) \right )
= \frac{n}{2} \left ( \delta_1 - \delta_2 \right ).
\end{align*}
If $G$ is convex (or linear), then
\begin{align*}
\sum_{i \in [n]} G(X_i)
\geq n G \left ( \frac{1}{n} \sum_{i \in [n]} X_i \right )
= n G \left (\frac{1}{2} \right )
=  \frac{n}{4} \left ( \delta_1 - \delta_2 \right ).
\end{align*}
Altogether, these observations imply that if $r \in \{1, 2\}$, then $\sum_{i \in [n]} G(X_i) \geq 0$ whenever $\delta_1 \geq \delta_2$.

\paragraph{Determining feasible $\lambda$.}
It remains to determine the range of $\lambda$ for which $\delta_1 \geq \delta_2$.
We now recall the precise specifications of $\delta_1$ and $\delta_2$:
\begin{equation*}
    \delta_1 = 1, \quad \delta_2 = \begin{cases} \lambda^p & \text{if } p = 1, \\ \frac{\lambda^p}{\left (1 - \lambda^{\frac{p}{p-1}} \right)^{p-1}} & \text{if } p > 1. \end{cases}
\end{equation*}
Since $p > 1$ (as we continue to assume $q > 1$), the following are equivalent assuming $\lambda \in (0, 1)$:
\begin{align*}
\delta_1 \geq \delta_2
&\iff \left (1 - \lambda^{\frac{p}{p-1}} \right)^{p-1} \geq \lambda^p \\
&\iff 1 - \lambda^{\frac{p}{p-1}} \geq \lambda^{\frac{p}{p-1}} \\
&\iff 2^{-1} \geq\lambda^{\frac{p}{p-1}} \\
&\iff 2^{1/p - 1}  \geq \lambda
\end{align*}
Since we assume $\lambda \in (0, 2^{1/p - 1}]$, it follows that $\delta_1 \geq \delta_2$ as desired.
\end{proof}

\subsubsection{Characterizing the Remaining Cases}
When $p/q \notin \{ 1, 2\}$, we are not able to immediately determine a range of values of feasible $\lambda$. Nevertheless, the structure of $G$ is still well-behaved, as we show in the following lemma.
The proof of this result, which follows from a standard calculus argument, is deferred to the appendix (\Cref{app:higher-d-second-derivative}).

\begin{lemma}\label{lem:higher-d-second-derivative}
The second derivative of $G$ on $(0, 1)$ is characterized as follows:
\begin{enumerate}
    \item If $p/q \in (1, 2)$, then there exists $Z^* \in (0, 1)$ such that $G''(Z) < 0$ for $Z \in (0, Z^*)$, $G''(Z) > 0$ for $Z \in (Z^*, 1)$, and $G''(Z^*) = 0$.
    \item If $p/q \in (0, 1) \cup (2, \infty)$, then there exists $Z^* \in (0, 1)$ such that $G''(Z) > 0$ for $Z \in (0, Z^*)$, $G''(Z) < 0$ for $Z \in (Z^*, 1)$, and $G''(Z^*) = 0$.
\end{enumerate}
\end{lemma}
Using \Cref{lem:higher-d-second-derivative}, we can significantly restrict the set of values in an optimal solution to \eqref{eq:relaxed-program-higher-d}.

\begin{claim}[Values of Optimal Solution]\label{lem:higher-d-optimal-solutions}
Let $X_1, \dots, X_n$ be an optimal solution to the program \eqref{eq:relaxed-program-higher-d}.
Let $Z^* \in (0, 1)$ be the inflection point guaranteed by \Cref{lem:higher-d-second-derivative}.
\begin{enumerate}
\item
If $p/q \in (1, 2)$, then there exist $a \in (0, Z^*), b \in [Z^*, 1]$ such that for all $i \in [n]$, $X_i \in \{0, a, b\}$. Moreover, $|\{ i \in [n] : X_i = a \}| \leq 1$.

\item
If $p/q \in (0, 1) \cup (2, \infty)$, then there exist $a \in [0, Z^*], b \in (Z^*, 1)$ such that for all $i \in [n]$, $X_i \in \{a, b, 1\}$. Moreover, $|\{ i \in [n] : X_i = b \}| \leq 1$.

\end{enumerate}
\end{claim}

\begin{proof}
Recall from \Cref{lem:higher-d-second-derivative} that the domain of $G$ is partitioned into regions of strict convexity and strict concavity. Specifically, if $p/q \in (0, 1) \cup (2, \infty)$, then $G$ is strictly convex on $[0, Z^*]$ and strictly concave on $[Z^*, 1]$. Conversely, if $p/q \in (1, 2)$, $G$ is strictly concave on $[0, Z^*]$ and strictly convex on $[Z^*, 1]$. Below, we reason about the case where $p/q \in (0, 1) \cup (2, \infty)$; the argument for $p/q \in (1, 2)$ follows analogously.

Suppose for the sake of contradiction that there exist distinct coordinates $i, j \in [n]$ such that $0 \leq X_i < X_j \leq Z^*$. By strict convexity, we have $G(X_i) + G(X_j) >  G(X_i + \epsilon) + G(X_j - \epsilon)$ for sufficiently small $\epsilon > 0$. Since this perturbation preserves the sum of the variables but decreases the objective value, the solution $X_1, \dots, X_n$ cannot be optimal, which is a contradiction. It follows that there exists $a \in [0, Z^*]$ such that for every $i \in [n]$, if $X_i \in [0, Z^*]$ then $X_i = a$.

Now suppose for the sake of contradiction that there exist distinct coordinates $i, j \in [n]$ such that $Z^* \leq X_i < X_j < 1$. By strict concavity, we have $G(X_i) + G(X_j) >  G(X_i - \epsilon) + G(X_j + \epsilon)$ for sufficiently small $\epsilon > 0$. Since this perturbation preserves the sum of the variables but decreases the objective value, the solution $X_1, \dots, X_n$ cannot be optimal, which is a contradiction. It follows that there can be at most one coordinate $i$ such that $X_i \in (Z^*, 1)$.

Observe that any remaining coordinates must therefore take the value 1, completing the proof.
\end{proof}

In the following sections, we use the above claim to solve the relaxed program.

\subsubsection{The Case of \texorpdfstring{$p/q \in (1, 2)$}{p/q in (1, 2)}}
Our analysis of this setting consists of two main steps. First, we use \Cref{lem:higher-d-optimal-solutions} to further simplify our optimization problem given in \eqref{eq:relaxed-program-higher-d}. Then we will solve the simplified optimization problem to find the appropriate value of $\lambda$.

\paragraph{Simplifying the Program in \eqref{eq:relaxed-program-higher-d}.}

Note that \Cref{lem:higher-d-optimal-solutions} significantly reduces the number of parameters in our optimization problem to just three, namely the values $a, b$ and $|\{ i \in [n] : X_i = a \} |$.
Indeed, as noted by \citet{Gravin:2025aa}, we can ignore the effect of this single point by taking $n\to\infty$: duplicating the underlying profile preserves the median constraints and scales the objective, while the contribution of one exceptional point is $O(1)$ since $G$ is continuous on $[0,1]$.
We therefore reduce the complexity of the program to a single parameter, namely the quantity $|\{ i \in [n] : X_i = b \}|$. Since we have $X_i \in \{0, b\}$ after this reduction, it follows that
\begin{equation*}
|\{ i \in [n] : X_i = 0 \}| = n - |\{ i \in [n] : X_i = b \}| .
\end{equation*}
Having reduced the optimization to a single parameter, we will minimize a simpler objective. Let $t =|\{ i \in [n] : X_i = b \}|$. Since $n/2 = \sum_{i\in [n]} X_i = tb$, it follows that
\[
|\{ i \in [n] : X_i = b \}| = \frac{n}{2b}, \quad |\{ i \in [n] : X_i = 0 \}| = \frac{n(2b - 1)}{2b}.
\]
Let $Z^*$ be the inflection point guaranteed by \Cref{lem:higher-d-second-derivative}.
The optimization problem has now been reduced to finding $b \in [Z^*, 1]$ such that
\begin{align}
&\min_{b\in [Z^*, 1]} \left [  \frac{n}{2b} \cdot G(b) + \frac{n(2b - 1)}{2b} \cdot G(0) \right ] = \min_{b\in [Z^*, 1]} \left [ \frac{n}{2b} \left ( G(b) + (2b - 1) \cdot G(0) \right ) \right ]
\geq 0
\end{align}
\begin{subequations}
    We recall from \eqref{eq:G-and-C-star-def} that the function $G$ is given by
\begin{equation}
    G(Z) = \delta_1 ( 1 - Z)^r - \delta_2 Z^r,
\end{equation}
where $r = p/q$ and $\delta_1, \delta_2 > 0$ are defined as follows:
\begin{equation}
    \delta_1 = 1, \quad \delta_2 = \begin{cases} \lambda^p & \text{if } p = 1, \\ \frac{\lambda^p}{\left (1 - \lambda^{\frac{p}{p-1}} \right)^{p-1}} & \text{if } p > 1. \end{cases}
\end{equation}
\end{subequations}
Since we assume that $q > 1$, and we also have $p/q \in (1, 2)$, only the second case of $\delta_2$ is relevant.
We want to find $\lambda \in (0, 1)$ such that the minimum of the function $u: [Z^*, 1] \to \mathbb R$ given by
\begin{align*}
u(b)
&= G(b) + (2b - 1) \cdot G(0)
\end{align*}
is zero.
Let $b^* \in \arg \min_{b \in [ Z^*, 1]} u(b)$ be a minimizer. In the following lemma, we show that if $u$ is equal to zero at $b^*$, then its derivative at $b^*$ is also zero. The proof of this lemma, which follows from standard calculus arguments, is deferred to \Cref{app:u-b-minimization}.
\begin{restatable}{lemma}{lemubminimization}
\label{lem:u-b-minimization}
If $u(b^*) = 0$, then $u'(b^*) = 0$.
\end{restatable}
Granting the above lemma and recalling that $G(0) = \delta_1 = 1$, we now find $\lambda$ such that
\begin{subequations}\label{eq:u-system-r-in-1-2}
\begin{align}
  u(b) &= G(b) + 2b - 1  = 0, \label{eq:u-b} \\
  u'(b) &= G'(b) + 2 = 0. \label{eq:u-b-derivative}
\end{align}
\end{subequations}

\paragraph{Finding $\lambda$.}
We perform the following steps to determine a feasible parameter $\lambda^*$.
\begin{enumerate}[label=(\roman*)]
    \item Let $b^*$ be a solution to the system of equations given in \eqref{eq:u-system-r-in-1-2}.
    \item From \eqref{eq:u-b}, we have $G(b^*) = 1 - 2b^*$. By the definition of $G(Z)$ given in \eqref{eq:G-and-C-star-def}, we also have
    $G(b^*) = (1 - b^*)^r - \delta_2 (b^*)^r.$
    Using these, we isolate $\delta_2$ to get
    \begin{equation}\label{eq:delta2-r-in-1-2}
        \delta_2 = \frac{(1 - b^*)^r - G(b^*)}{(b^*)^r} = \frac{(1 - b^*)^r + 2b^* - 1}{(b^*)^r}
    \end{equation}
    Note that for $r \in (1, 2)$, the numerator $h(b) = (1-b)^r + 2b - 1$ has derivative $h'(b) = 2 - r(1-b)^{r-1} \geq 2 - r > 0$. Since $h(0) = 0$ and $h$ is strictly increasing, it follows that the numerator is positive for any $b^* > 0$. Therefore, we are guaranteed that $\delta_2 > 0$.
    \item We set $\lambda^*$ as follows:
    \begin{equation}\label{eq:lambda-star-r-in-1-2}
        \lambda^* =
            \delta_2^{\frac{1}{p}} \cdot  \left ( 1 + \delta_2^{\frac{1}{p-1}} \right )^{\frac{1- p}{p}}.
    \end{equation}
    We obtained the value of $\lambda^*$ by observing the following equivalences.
    \begin{align*}
    \delta_2 = \frac{\lambda^p}{\left (1 - \lambda^{\frac{p}{p-1}} \right)^{p-1}}
    &\iff \delta_2 \cdot \left (1 - \lambda^{\frac{p}{p-1}} \right)^{p-1} = \lambda^p \\
    &\iff \delta_2^{\frac{1}{p-1}} \cdot \left (1 - \lambda^{\frac{p}{p-1}} \right )= \lambda^{\frac{p}{p-1}} \\
    &\iff \delta_2^{\frac{1}{p-1}} \cdot  \left ( 1 + \delta_2^{\frac{1}{p-1}} \right )^{-1} =\lambda^{\frac{p}{p-1}} \\
    &\iff \delta_2^{\frac{1}{p}} \cdot  \left ( 1 + \delta_2^{\frac{1}{p-1}} \right )^{\frac{1-p}{p}} =\lambda
    \end{align*}
    \item Finally we set $\mathtt{UB}(p, q) = 1/ \lambda^*$. Note that this is well-defined since $1 + \delta_2^{\frac{1}{p-1}} > \delta_2^{\frac{1}{p-1}} > 0$.
\end{enumerate}

\begin{proposition}\label{prop:p-over-q-in-1-2}
If $p/q \in (1, 2)$ then $\mathtt{UB}(p, q) \leq 1 + 2^{1/p}$.
\end{proposition}
\begin{proof}
By definition, we have
\begin{equation*}
\mathtt{UB}(p, q)  = \delta_2^{-\frac{1}{p}} \left ( 1 + \delta_2^{\frac{1}{p-1}} \right )^{\frac{p-1}{p}} = \left( 1 + \delta_2^{-\frac{1}{p-1}} \right)^{\frac{p-1}{p}}.
\end{equation*}
To upper bound $\mathtt{UB}(p, q)$, it suffices to find a lower bound for $\delta_2$.
Note that if $\delta_2 \geq 1/2$, then
\begin{equation*}
    \mathtt{UB}(p, q) \leq \left( 1 + 2^{\frac{1}{p-1}} \right)^{\frac{p-1}{p}} \leq  1 + 2^{\frac{1}{p}},
\end{equation*}
which implies our desired result. In the final inequality, we use the fact that $(A +B)^z \le A^z + B^z$ for $z \in (0, 1)$ and $A, B \geq 0$. The rest of the proof is devoted to establishing a lower bound on $\delta_2$.

\paragraph{Lower bounding $\delta_2$.} Let $r = p/q$.
From equation \eqref{eq:u-b}, (i.e., setting $u(b^*) = 0$), we have
\begin{equation*}
    \delta_2 = \frac{(1 - b^*)^r + 2b^* - 1}{(b^*)^r}.
\end{equation*}
This is precisely equation \eqref{eq:delta2-r-in-1-2}.
By additionally enforcing \eqref{eq:u-b-derivative} (i.e., $u'(b^*) = 0$), we can extract another expression for $\delta_2$. Since $G'(b^*) = -r  (1 - b^*)^{r-1} - r \delta_2 (b^*)^{r-1} $ and $G'(b^*) + 2 = 0$, we have
\begin{equation*}
    \delta_2 = \frac{2 - r(1-b^*)^{r-1}}{r (b^*)^{r-1}}.
\end{equation*}
By equating the two above expressions for $\delta_2$, we observe the following equivalences:
\begin{align*}
    \frac{(1 - b^*)^r + 2b^* - 1}{(b^*)^r} = \frac{2 - r(1-b^*)^{r-1}}{r (b^*)^{r-1}}
    &\iff r \left ( (1 - b^*)^r + 2b^* - 1 \right ) = b^* \left( 2 - r(1-b^*)^{r-1} \right) \\
    &\iff r (1 - b^*)^r + r (2b^* - 1) = 2b^* - r b^* (1-b^*)^{r-1} \\
    &\iff r(1 - b^*)^r + r b^*(1-b^*)^{r-1} = 2b^* - r(2b^* - 1) \\
    &\iff r(1 - b^*)^{r-1}(1 - b^* + b^*) = r - 2b^*(r - 1) \\
    &\iff (1 - b^*)^{r-1} = 1 - 2b^* \left( \frac{r-1}{r} \right).
\end{align*}
This yields the identity
\begin{equation}\label{eq:b-star-identity}
(1-b^*)^{r-1} = 1 - 2b^* \left(\frac{r-1}{r}\right).
\end{equation}
Consider the difference of the two sides of \eqref{eq:b-star-identity}, given by
\begin{equation*}
h(b) = (1-b)^{r-1} - \left (1 - 2b  \left ( \frac{r-1}{r} \right ) \right ).
\end{equation*}
Its derivative is
\begin{equation*}
h'(b) = (r-1) \left [ \frac{2}{r} - (1-b)^{r-2} \right ].
\end{equation*}
Because $r \in (1, 2)$, the exponent $r-2$ is negative. Thus, the map $b \mapsto (1-b)^{r-2}$ is strictly increasing on the domain $(0, 1)$. This means that $h'(b)$ is strictly decreasing on the domain $(0, 1)$.

Since $h'(b)$ is strictly decreasing on the open interval $(0, 1)$, $h(b)$ is strictly concave on $[0, 1]$. Since $h(0) = 0$ and $h'(0) = (r-1)(2/r - 1) > 0$, strict concavity dictates that $h(b)$ increases to a unique positive maximum before strictly decreasing, crossing zero exactly once on $(0, 1]$ at $b^*$.
Evaluating at $b = 1/2$ yields $h(1/2) = 2^{1-r} - 1/r$. Because $2^{r-1} \leq r$ for all $r \in [1, 2]$, we have $2^{1-r} \geq 1/r$, meaning $h(1/2) \geq 0$. Since $h(1) = \frac{r-2}{r} < 0$, the root $b^*$ must lie in the interval $[1/2, 1]$.

Finally, we are ready to bound $\delta_2$. By substituting the identity from \eqref{eq:b-star-identity} back into the second expression for $\delta_2$, we eliminate the exponential term $(1-b^*)^{r-1}$ to get
\begin{equation*}
\delta_2 = \frac{2 - r\left(1 - 2b^* \left (\frac{r-1}{r}\right)\right)}{r (b^*)^{r-1}} = \frac{2 - r + 2b^*(r-1)}{r (b^*)^{r-1}}.
\end{equation*}
Since $b^* \geq 1/2$, the numerator is bounded below by $2 - r + 2(1/2)(r - 1) = 1$. Since $b^* \leq 1$ and $r \in (1, 2)$, the denominator is bounded above by $2(1)^{r-1} = 2$. Thus, $\delta_2 \geq 1/2$ as desired.
\end{proof}

\subsubsection{The Case of \texorpdfstring{$p/q \in (2, \infty)$}{r in (2, infty)}}
We next consider the case of $p/q \in  (2, \infty)$. Our analysis continues in a similar fashion.

\paragraph{Simplifying the optimization problem.}
\Cref{lem:higher-d-optimal-solutions} significantly reduces the number of parameters in our optimization problem to just three, namely the values $a, b$ and $|\{ i \in [n] : X_i = b \} |$. We again ignore the effect of the single point with $X_i = b$. We thus reduce the optimization problem to a single parameter, namely the quantity $|\{ i \in [n] : X_i = a \}|$. Note that
\begin{equation*}
|\{ i \in [n] : X_i = 1 \}| = n - |\{ i \in [n] : X_i = a \}| .
\end{equation*}
Having reduced the optimization to a single parameter, we will minimize a simpler objective. Let $t =|\{ i \in [n] : X_i = a \}|$. Since $n/2 = \sum_{i\in [n]} X_i = ta + (n-t) = t(a - 1) + n$, it follows that
\[
|\{ i \in [n] : X_i = a \}| = \frac{n}{2(1-a)}, \quad |\{ i \in [n] : X_i = 1 \}| = \frac{n(1 - 2a)}{2(1-a)}.
\]
Note that since we require that $|\{ i \in [n] : X_i = 1 \}| \ge 0$, it must be the case that $a \le 1/2$.
Let $Z^*$ be the inflection point guaranteed by \Cref{lem:higher-d-second-derivative}.
The optimization problem has now been reduced to finding $a \in [0, \min(Z^*, 1/2)]$ such that
\begin{align*}
&\min_{a \in [0, \min(Z^*, 1/2)]} \left [  \frac{n}{2(1-a)} \cdot G(a) + \frac{n(1 - 2a)}{2(1-a)} \cdot G(1) \right ]  \\
&= \min_{a \in [0, \min(Z^*, 1/2)]} \left [  \frac{n}{2(1-a)} \left ( G(a) + (1 - 2a) \cdot G(1) \right ) \right ]
\geq 0
\end{align*}
\begin{subequations}
    We recall from \eqref{eq:G-and-C-star-def} that the function $G$ is given by
\begin{equation}
    G(Z) = \delta_1 ( 1 - Z)^r - \delta_2 Z^r,
\end{equation}
where $r = p/q$ and $\delta_1, \delta_2 > 0$ are defined as follows:
\begin{equation}
    \delta_1 = 1, \quad \delta_2 = \begin{cases} \lambda^p & \text{if } p = 1, \\ \frac{\lambda^p}{\left (1 - \lambda^{\frac{p}{p-1}} \right)^{p-1}} & \text{if } p > 1. \end{cases}
\end{equation}
\end{subequations}
Since we assume that $q > 1$, and we also have $p/q \in (2, \infty)$, only the second case of $\delta_2$ is relevant.
Now, our goal is to find $\lambda \in (0, 1)$ such that the minimum of the function $u: [0,Z^*] \to \mathbb R$ given by
\begin{align*}
u(a)
&=  G(a) + (1 - 2a) \cdot G(1)
\end{align*}
is zero.
Let $a^* \in \arg \min_{a \in [0, \min(Z^*, 1/2)]} u(a)$ be a minimizer. In the following lemma, we show that if $u$ is zero at $a^*$, then its derivative at $a^*$ is also zero.
We defer the proof of this lemma to the appendix (\Cref{app:u-a-minimization}).

\begin{restatable}{lemma}{lemuaminimizationrintwoinf}
\label{lem:u-a-minimization-r-in-2-inf}
If $u(a^*) = 0$, then $u'(a^*) = 0$ and $a^* < 1/2$.
\end{restatable}

Granting this lemma, it remains to find $\lambda$ such that
\begin{subequations}\label{eq:u-system-r-in-2-inf}
\begin{align}
 u(a) &= G(a) + (1 - 2a) \cdot G(1)  = 0, \label{eq:u-a-r-in-2-inf} \\
 u'(a) &= G'(a) - 2 \cdot G(1) = 0. \label{eq:u-a-derivative-r-in-2-inf}
\end{align}
\end{subequations}

\paragraph{Finding a feasible $\lambda$.}
We perform the following steps to determine a feasible parameter $\lambda^*$.
\begin{enumerate}[label=(\roman*)]
    \item Let $a^*$ be a solution to the system of equations given in \eqref{eq:u-system-r-in-2-inf}.
    \item By definition of $G$, we have $G(1) = -\delta_2$. Thus, by \eqref{eq:u-a-r-in-2-inf} we have $G(a^*) = (1 - 2a ^*)\delta_2.$
    Since $a^* < 1/2$, we divide by $1 - 2a^*$, and isolate $\delta_2$:
    \begin{align*}
        \delta_2 = \frac{G(a^*)}{1 - 2a^*} &= \frac{\delta_1(1 - a^*)^r - \delta_2 (a^*)^r}{1 - 2a^*}.
    \end{align*}
    Note that this simplifies to
    \begin{equation}\label{eq:delta2-r-in-2-inf}
    \delta_2 = \frac{\delta_1(1 - a^*)^r}{1 - 2a^* + (a^*)^r}.
    \end{equation}
    Because $1-2a^* > 0$ and $(a^*)^r \geq 0$, the denominator $1 - 2a^* + (a^*)^r$ is strictly positive, so this expression for $\delta_2$ is well-defined.

    We also note that the map $a \mapsto \delta_1(1-a)^r$ is strictly decreasing for $a \in [0, 1/2]$ and is therefore bounded below by $\delta_1(1/2)^r > 0$. Since we have already established that the denominator is strictly positive, it follows that $\delta_2 > 0$.

    \item We set $\lambda^*$ as follows:
    \begin{equation*}
        \lambda^* = \delta_2^{\frac{1}{p}} \cdot  \left ( 1 + \delta_2^{\frac{1}{p-1}} \right )^{\frac{1-p}{p}}
    \end{equation*}
    We obtained this value of $\lambda^*$ by observing the following equivalences.
    \begin{align*}
    \delta_2 = \frac{\lambda^p}{\left (1 - \lambda^{\frac{p}{p-1}} \right)^{p-1}}
    &\iff \delta_2 \cdot \left (1 - \lambda^{\frac{p}{p-1}} \right)^{p-1} = \lambda^p \\
    &\iff \delta_2^{\frac{1}{p-1}} \cdot \left (1 - \lambda^{\frac{p}{p-1}} \right )= \lambda^{\frac{p}{p-1}} \\
    &\iff \delta_2^{\frac{1}{p-1}} \cdot  \left ( 1 + \delta_2^{\frac{1}{p-1}} \right )^{-1} =\lambda^{\frac{p}{p-1}} \\
    &\iff \delta_2^{\frac{1}{p}} \cdot  \left ( 1 + \delta_2^{\frac{1}{p-1}} \right )^{\frac{1- p}{p}} =\lambda
    \end{align*}

    \item Finally we set $\mathtt{UB}(p, q) = 1/ \lambda^*$. Note that this is well-defined because $1 + \delta_2^{\frac{1}{p-1}} > \delta_2^{\frac{1}{p}} > 0$.
\end{enumerate}

\begin{proposition}\label{prop:p-over-q-in-2-inf}
If $p/q \in (2, \infty)$ then $\mathtt{UB}(p, q) \leq 1 + 2^{1/q}$.
\end{proposition}
\begin{proof}
By definition, we have
\begin{equation*}
\mathtt{UB}(p, q)  = \delta_2^{-\frac{1}{p}} \left ( 1 + \delta_2^{\frac{1}{p-1}} \right )^{\frac{p-1}{p}} = \left( 1 + \delta_2^{-\frac{1}{p-1}} \right)^{\frac{p-1}{p}}.
\end{equation*}
To upper bound $\mathtt{UB}(p, q)$, it suffices to find a lower bound for $\delta_2$.
Note that if $\delta_2 \geq 2^{-p/q}$, then
\begin{equation*}
    \mathtt{UB}(p, q) \leq \left( 1 + (2^{-p/q})^{-\frac{1}{p-1}} \right)^{\frac{p-1}{p}} = \left( 1 + 2^{\frac{p/q}{p-1}} \right)^{\frac{p-1}{p}} \leq  1 + 2^{\frac{p/q}{p}} = 1 + 2^{1/q},
\end{equation*}
which implies our desired result. In the final inequality, we use the fact that $(A + B)^z \le A^z + B^z$ for $z \in (0, 1)$ and $A, B \geq 0$. The rest of the proof is devoted to establishing a lower bound on $\delta_2$.

\paragraph{Lower bounding $\delta_2$.} Let $r = p/q$.
From equation \eqref{eq:u-a-r-in-2-inf}, (i.e., setting $u(a^*) = 0$), we isolated $\delta_2$ as
\begin{equation*}
    \delta_2 = \frac{(1 - a^*)^r}{1 - 2a^* + (a^*)^r}.
\end{equation*}
This is precisely equation \eqref{eq:delta2-r-in-2-inf} (recalling that $\delta_1 = 1$).
By additionally enforcing \eqref{eq:u-a-derivative-r-in-2-inf} (i.e., $u'(a^*) = 0$), we can extract another expression for $\delta_2$. Since $G'(a^*) = -r(1 - a^*)^{r-1} - r \delta_2 (a^*)^{r-1}$ and $G'(a^*) + 2\delta_2 = 0$, we have
\begin{equation*}
    0 = -r(1-a^*)^{r-1} - r\delta_2(a^*)^{r-1} + 2\delta_2 \implies \delta_2(2 - r(a^*)^{r-1}) = r(1-a^*)^{r-1}.
\end{equation*}
We quickly verify that $2 - r(a^*)^{r-1} \neq 0$. If the term were zero, the right-hand side would require $r(1-a^*)^{r-1} = 0$, which is impossible since $r > 0$ and $a^* < 1/2$. Thus, we can safely divide to obtain
\begin{equation*}
    \delta_2 = \frac{r(1-a^*)^{r-1}}{2 - r (a^*)^{r-1}}.
\end{equation*}

Now, by equating the two above expressions for $\delta_2$, we observe the following equivalences:
\begin{align*}
    \frac{(1 - a^*)^r}{1 - 2a^* + (a^*)^r} = \frac{r(1-a^*)^{r-1}}{2 - r (a^*)^{r-1}}
    &\iff (1 - a^*) \left ( 2 - r(a^*)^{r-1} \right ) = r \left( 1 - 2a^* + (a^*)^r \right) \\
    &\iff 2(1 - a^*) - r(a^*)^{r-1}(1-a^*) = r(1 - 2a^*) + r(a^*)^r \\
    &\iff 2(1 - a^*) - r(a^*)^{r-1} + r(a^*)^r = r(1 - 2a^*) + r(a^*)^r \\
    &\iff r(a^*)^{r-1} = 2(1 - a^*) - r(1 - 2a^*) \\
    &\iff r(a^*)^{r-1} = 2 - 2a^* - r + 2ra^* \\
    &\iff r(a^*)^{r-1} = 2 - r + 2a^*(r - 1).
\end{align*}
This yields the identity
\begin{equation}\label{eq:a-star-identity}
(a^*)^{r-1} = \frac{2-r}{r} + 2a^* \left( \frac{r-1}{r} \right).
\end{equation}

Finally, we are ready to bound $\delta_2$. By substituting the identity from \eqref{eq:a-star-identity} back into the second expression for $\delta_2$, we eliminate the term $r(a^*)^{r-1}$ in the denominator to get
\begin{equation*}
\delta_2 = \frac{r(1-a^*)^{r-1}}{2 - \left( 2 - r + 2a^*(r-1) \right)} = \frac{r (1-a^*)^{r-1}}{r - 2a^*(r-1)}.
\end{equation*}
By \Cref{lem:u-a-minimization-r-in-2-inf}, we know that $a^* \in [0, 1/2)$. Therefore, $1 - a^* > 1/2$, which means the numerator is strictly bounded below by $r(1/2)^{r-1} = r 2^{1-r}$. Additionally, since $a^* \geq 0$ and $r > 2$, the denominator is bounded above by $r$.
Thus, we conclude that
\begin{equation*}
\delta_2 > \frac{r 2^{1-r}}{r} = 2^{1-r} > 2^{-r}.
\end{equation*}
This yields $\delta_2 \geq 2^{-r}$ as desired.
\end{proof}
\clearpage

\subsubsection{The Case of \texorpdfstring{$p/q \in (0, 1)$}{r in (0,1) }}

Finally, we consider the case of $p/q \in (0, 1)$. Our analysis resembles that of previous sections, but we emphasize key distinctions as they arise.

\paragraph{Simplifying the optimization problem.}
\Cref{lem:higher-d-optimal-solutions} significantly reduces the number of parameters in our optimization problem to just three, namely the values $a, b$ and $|\{ i \in [n] : X_i = b \} |$.
By an argument similar to the one in the case of $p/q \in (2, \infty)$, our problem reduces to finding $a \in [0, \min(Z^*, 1/2)]$ (where $Z^*$ is the inflection point guaranteed by \Cref{lem:higher-d-second-derivative}) such that
\begin{align*}
&\min_{a \in [0, \min(Z^*, 1/2)]} \left [  \frac{n}{2(1-a)} \cdot G(a) + \frac{n(1 - 2a)}{2(1-a)} \cdot G(1) \right ]  \\
&= \min_{a \in [0, \min(Z^*, 1/2)]} \left [  \frac{n}{2(1-a)} \left ( G(a) + (1 - 2a) \cdot G(1) \right ) \right ]
\geq 0
\end{align*}
This time, the function $G$ as defined in \eqref{eq:G-and-C-star-def} is given by
\begin{subequations}
\begin{equation}
    G(Z) = \delta_1 ( 1 - Z)^r - \delta_2 Z^r,
\end{equation}
where $r = p/q$ and $\delta_1, \delta_2 > 0$ are defined as follows:
\begin{equation}\label{eq:delta-def-zero-one}
    \delta_1 = \left (1 - \gamma^{\frac{p}{1 - p/q}}\right )^{1 - p/q},
    \quad
    \delta_2 = \begin{cases}
        \lambda^p & \text{if } p = 1, \\
        \frac{\lambda^p}{\left (1 - \left (\frac{\lambda}{\gamma}\right )^{\frac{p}{p-1}} \right)^{p-1}} & \text{if } 1 < p < q
    \end{cases}
\end{equation}
\end{subequations}
Now, our goal is to find $\lambda \in (0, 1)$ and $\gamma \in (\lambda, 1)$ such that the minimum of the function $u: [0,Z^*] \to \mathbb R$ given by
\begin{align*}
u(a)
&=  G(a) + (1 - 2a) \cdot G(1)
\end{align*}
is zero.
Let $a^* \in \arg \min_{a \in [0, \min(Z^*, 1/2)]} u(a)$ be a minimizer. We can again show that if $u$ is zero at $a^*$, then its derivative at $a^*$ is also zero. A proof of this lemma is deferred to the appendix (\Cref{app:u-a-minimization}).

\begin{restatable}{lemma}{lemuaminimizationrinzeroone}
\label{lem:u-a-minimization-r-in-zero-one}
If $u(a^*) = 0$, then $u'(a^*) = 0$ and $a^* < 1/2$.
\end{restatable}

Granting this lemma, it remains to find $\lambda$ such that
\begin{subequations}\label{eq:u-system-r-in-zero-one}
\begin{align}
 u(a) &= G(a) + (1 - 2a) \cdot G(1)  = 0, \label{eq:u-a-r-in-zero-one} \\
 u'(a) &= G'(a) - 2 \cdot G(1) = 0. \label{eq:u-a-derivative-r-in-zero-one}
\end{align}
\end{subequations}

\paragraph{Finding $\lambda$ and $\gamma$.}
\begin{enumerate}[label=(\roman*)]
    \item Let $a^*$ be a solution to the system of equations given in \eqref{eq:u-system-r-in-zero-one}.
    \item By definition of $G$, we have $G(1) = -\delta_2$. Thus, by \eqref{eq:u-a-r-in-zero-one} we have $G(a^*) = (1 - 2a ^*)\delta_2.$
    Since $a^* < 1/2$, we divide by $1 - 2a^*$ and isolate $\delta_2$:
    \begin{align*}
        \delta_1(1 - a^*)^r - \delta_2 (a^*)^r = (1 - 2a^*) \delta_2.
    \end{align*}
    This implies that
    \begin{equation}\label{eq:delta2-r-in-zero-one}
    \delta_2 = \frac{\delta_1(1 - a^*)^r}{1 - 2a^* + (a^*)^r} .
    \end{equation}
    Note that this step is identical to the one taken in the case of $p/q \in (2, \infty)$. We remind the reader that $\delta_2 > 0$.
    \item We set the parameter $\gamma \in (0, 1)$ such that
    \begin{equation}\label{eq:gamma-def-zero-one}
        \gamma = \left( 1 + 2^{1/p} \right)^{-1/p}.
    \end{equation}
    This allows us to treat $\delta_1 = \left (1 - \gamma^{\frac{p}{1 - p/q}}\right )^{1 - p/q}$ as a fixed constant.
    \item Next, we extract a feasible $\lambda$ strictly in terms of $\delta_2$ and $\gamma$. Since we continue to assume $q>1$, the case $p=1$ lies in the regime $1\leq p<q$ and has $C^*=0$, so $\lambda^*=\delta_2$. For $1 < p < q$, we define $\lambda$ to be the unique positive solution to
    \begin{equation*}
        \delta_2 = \frac{\lambda^p}{\left (1 - \left (\frac{\lambda}{\gamma}\right )^{\frac{p}{p-1}} \right)^{p-1}}.
    \end{equation*}
    By isolating $\lambda$, we observe the following equivalences:
    \begin{align*}
    \delta_2 \cdot \left (1 - \lambda^{\frac{p}{p-1}} \gamma^{-\frac{p}{p-1}} \right)^{p-1} = \lambda^p
    &\iff \delta_2^{\frac{1}{p-1}} \cdot \left (1 - \lambda^{\frac{p}{p-1}} \gamma^{-\frac{p}{p-1}} \right )= \lambda^{\frac{p}{p-1}} \\
    &\iff \delta_2^{\frac{1}{p-1}} = \lambda^{\frac{p}{p-1}} \left( 1 + \delta_2^{\frac{1}{p-1}} \gamma^{-\frac{p}{p-1}} \right) \\
    &\iff \lambda^{\frac{p}{p-1}} = \delta_2^{\frac{1}{p-1}} \cdot \left( 1 + \delta_2^{\frac{1}{p-1}} \gamma^{-\frac{p}{p-1}} \right)^{-1} \\
    &\iff \lambda^* = \delta_2^{\frac{1}{p}} \cdot \left( 1 + \delta_2^{\frac{1}{p-1}} \gamma^{-\frac{p}{p-1}} \right)^{\frac{1-p}{p}}.
    \end{align*}
    \item Finally we set $\mathtt{UB}(p, q) = 1/ \lambda^*$. Note that this is well-defined because $\delta_2 > 0$.
\end{enumerate}

\begin{proposition}\label{prop:p-over-q-in-0-1}
If $p/q \in (0, 1)$ then $\mathtt{UB}(p, q) \leq 1 + 2^{1/p}$.
\end{proposition}
\begin{proof}
By definition, we have
\begin{equation*}
    \mathtt{UB}(p, q) =
    \begin{cases}
        1 / \delta_2 & \text{if } p = 1, \\
        \delta_2^{-\frac{1}{p}} \left( 1 + \delta_2^{\frac{1}{p-1}} \gamma^{-\frac{p}{p-1}} \right)^{\frac{p-1}{p}} & \text{if } 1 < p < q.
    \end{cases}
\end{equation*}
Note that if $\delta_2 \geq \frac{1}{2 + 2^{1-1/p}}$, then for $1 < p < q$, we can upper bound $\mathtt{UB}(p, q)$ as follows:
\begin{align*}
    \mathtt{UB}(p, q) &\leq \left( \left( 2 + 2^{1-1/p} \right)^{\frac{1}{p-1}} + \gamma^{-\frac{p}{p-1}} \right)^{\frac{p-1}{p}}.
\end{align*}
Since we defined $\gamma = (1 + 2^{1/p})^{-1/p}$, we have $\gamma^{-\frac{p}{p-1}} = (1 + 2^{1/p})^{\frac{1}{p-1}}$. By substituting this in and factoring it out, the upper bound evaluates directly:
\begin{align*}
    \mathtt{UB}(p, q) &\leq \left( 2^{1/p}\left( 1 + 2^{1/p} \right)^{\frac{1}{p-1}} + \left( 1 + 2^{1/p} \right)^{\frac{1}{p-1}} \right)^{\frac{p-1}{p}} \\
    &= \left( \left( 1 + 2^{1/p} \right)^{\frac{1}{p-1}} \left( 2^{1/p} + 1 \right) \right)^{\frac{p-1}{p}} \\
    &= \left( \left( 1 + 2^{1/p} \right)^{\frac{p}{p-1}} \right)^{\frac{p-1}{p}} \\
    &= 1 + 2^{1/p}.
\end{align*}
Notice that the same claim holds for $p=1$ in this regime: since we continue to assume $q>1$, we have $r=1/q\in(0,1)$ and $\delta_2=\lambda$. Substituting $p=1$ yields $\delta_2 \geq 1/3$, which guarantees $\mathtt{UB}(1, q) = 1/\delta_2 \leq 3 = 1 + 2^{1/1}$. The rest of the proof is devoted to establishing this lower bound on $\delta_2$.

\paragraph{Lower bounding $\delta_2$.} Let $r = p/q$.
From equation \eqref{eq:u-a-r-in-zero-one}, we isolated $\delta_2$ as
\begin{equation*}
    \delta_2 = \frac{\delta_1(1 - a^*)^r}{1 - 2a^* + (a^*)^r}.
\end{equation*}
This is precisely equation \eqref{eq:delta2-r-in-zero-one}.
By also enforcing \eqref{eq:u-a-derivative-r-in-zero-one} (i.e., $u'(a^*) = 0$), we extract a second expression for $\delta_2$. Since $G'(a^*) = -r\delta_1(1 - a^*)^{r-1} - r \delta_2 (a^*)^{r-1}$ and $G'(a^*) + 2\delta_2 = 0$, we have
\begin{equation*}
    0 = -r\delta_1(1 - a^*)^{r-1} - r \delta_2 (a^*)^{r-1} + 2\delta_2 \implies \delta_2(2 - r(a^*)^{r-1}) = r\delta_1(1-a^*)^{r-1}.
\end{equation*}
As in the previous case, we verify that $2 - r(a^*)^{r-1} \neq 0$. If it were zero, this would require $r\delta_1(1-a^*)^{r-1} = 0$, which is impossible since $r, \delta_1 > 0$ and $a^* < 1/2$. Thus, we can divide to obtain
\begin{equation*}
    \delta_2 = \frac{r\delta_1(1-a^*)^{r-1}}{2 - r (a^*)^{r-1}}.
\end{equation*}
Equating the two expressions for $\delta_2$ produces the exact same identity for $(a^*)^{r-1}$ as before:
\begin{equation}\label{eq:a-star-identity-zero-one}
(a^*)^{r-1} = \frac{2-r}{r} + 2a^* \left( \frac{r-1}{r} \right).
\end{equation}
Notice that multiplying by $r$ and rearranging this identity yields $2 - r(a^*)^{r-1} - 2a^*(1-r) = r$. Since $r \in (0, 1)$, we know $1-r \in (0, 1)$. Applying Bernoulli's inequality to the concave term $(1-a^*)^{1-r}$, we have $(1-a^*)^{r-1} \geq \frac{1}{1 - a^*(1-r)}$. Applying this directly to our second expression for $\delta_2$ allows us to cleanly isolate and lower bound $\delta_2$:
\begin{align*}
    \delta_2 &= \frac{r\delta_1(1-a^*)^{r-1}}{2 - r(a^*)^{r-1}} \\
    &\geq \frac{r\delta_1}{\big(1 - a^*(1-r)\big)\big(2 - r(a^*)^{r-1}\big)} \tag{Bernoulli's inequality} \\
    &= \frac{r\delta_1}{2 - r(a^*)^{r-1} - 2a^*(1-r) + r(1-r)(a^*)^r}  \\
    &= \frac{r\delta_1}{r + r(1-r)(a^*)^r} \tag{by \eqref{eq:a-star-identity-zero-one}} \\
    &= \frac{\delta_1}{1 + (1-r)(a^*)^r}.
\end{align*}
We now bound the denominator of this final expression. Since $1-r < 1$ and $(a^*)^r < 1$ (by \Cref{lem:u-a-minimization-r-in-zero-one}), the term $(1-r)(a^*)^r$ is strictly less than $1$. Therefore, the denominator is strictly bounded above by $2$, yielding
$\delta_2 > \delta_1/2$.

It remains to lower bound the constant $\delta_1$.
Recall we defined $\delta_1 = \left (1 - \gamma^{\frac{p}{1 - r}}\right )^{1 - r}$ in \eqref{eq:delta-def-zero-one} and $\gamma^p = (1 + 2^{1/p})^{-1}$ in \eqref{eq:gamma-def-zero-one}. Since $r \in (0, 1)$, we have $\frac{p}{1-r} > p$. Because $\gamma \in (0, 1)$, this implies $\gamma^{\frac{p}{1-r}} < \gamma^p$, which in turn gives $1 - \gamma^{\frac{p}{1-r}} > 1 - \gamma^p$.
Furthermore, raising a value in $(0, 1)$ to a power $1-r \in (0, 1)$ strictly increases it. Thus, we have
\begin{equation*}
\delta_1 > (1 - \gamma^p)^{1-r} > 1 - \gamma^p.
\end{equation*}
Substituting $\gamma^p = \frac{1}{1 + 2^{1/p}}$, we further note that
\begin{equation*}
1 - \gamma^p = 1 - \frac{1}{1 + 2^{1/p}} = \frac{2^{1/p}}{1 + 2^{1/p}} = \frac{2}{2 + 2^{1-1/p}}.
\end{equation*}
Therefore, $\delta_1 > \frac{2}{2 + 2^{1-1/p}}$. Since $\delta_2 \geq \delta_1/2$, this implies our desired result.
\end{proof}

\subsubsection{Proof of \texorpdfstring{\Cref{thm:cm-upper-bound}}{General Result}} \label{app:high-dim-thm1}
Finally, we are ready to establish the general upper bound.

\begin{proof}[Proof of \Cref{thm:cm-upper-bound}]
We apply  \Cref{prop:p-over-q-is-1-or-2}, \Cref{prop:p-over-q-in-1-2}, \Cref{prop:p-over-q-in-0-1}, and \Cref{prop:p-over-q-in-2-inf} to
establish the upper bounds for all the cases involving $p$ and $q$ with $q>1$. The case $q = 1$ follows from a limiting argument.
To complete the proof, note that for all regimes, our upper bound on the approximation ratio is at most $1 + 2^{1/\min(p,q)}$. Since $p, q \geq 1$, it follows that
\begin{equation*}
\alpha_p(\M) \leq \mathtt{UB}(p, q) \leq 1 + 2^{1/ \min(p,q)} \leq 3,
\end{equation*}
as desired.
\end{proof}

%% file: conclusion.tex
\subsection{Discussion and Future Directions}
Our work illustrates that the coordinate-wise median mechanism is highly effective at providing an approximately optimal facility location while incentivizing agents to be truthful.
We provide its exact approximation ratio in $\ell_q(\mathbb R^2)$ under $p$-norm social costs, and show that its approximation ratio is at most $3$ in $\ell_q(\mathbb R^d)$ even for arbitrary monotone symmetric norm objectives.
There are several directions for future work and open questions. We outline them below.

Is the coordinate-wise median mechanism optimal among deterministic strategyproof mechanisms for $\ell_2(\mathbb R^d)$ under an arbitrary $p$-norm social cost? This is indeed the case when $d \in \{1, 2\}$ \citep{Feigenbaum:2017aa,Goel:2023aa}, but our understanding beyond this is considerably limited. While there are well-known characterizations for strategyproof mechanisms in $\mathbb R^2$ \citep{border1983straightforward, peters1993range}, we do not have similarly crisp characterizations for higher dimensions. New tools may need to be developed for the more general setting.

Another enticing direction is determining the extent to which randomized strategyproof-in-expectation mechanisms can achieve better approximation ratios than deterministic strategyproof mechanisms. \citet{barak2026facility} recently introduced a randomized strategyproof mechanism with a better approximation ratio than the CM mechanism in $\mathbb R^2$ for the utilitarian social cost.
For each $p$-norm social cost, does there exist a randomized strategyproof-in-expectation mechanism with a better approximation ratio than the CM mechanism in $\mathbb R^2$? Is there a single randomized strategyproof-in-expectation mechanism that \textit{simultaneously} achieves a better approximation for all $p$-norm social costs in $\mathbb R^2$? Does the same hold in $\mathbb R^d$ for $d \geq 3$?

\subsection{Acknowledgements}
JH thanks Xizhi Tan for an introduction to the strategic facility location problem and suggesting relevant literature.
JH also thanks Prasanna Ramakrishnan for helpful discussions and feedback on an earlier version of this manuscript. JH also thanks an anonymous reviewer for helpful comments on an earlier version of this manuscript.

\paragraph{AI disclosure}
JH initially used Gemini 3.1 Pro (High) in Antigravity to brainstorm and refine conceptual approaches, and to rephrase sections of this manuscript for clarity and flow. In particular,
\Cref{lem:min-sum-opp} and the limiting analysis used to establish \Cref{thm:cm-upper-bound} were significantly streamlined with the help of the model. 

For the second version of this work, JH also used GPT 5.5 (Extra High) in Codex to explore extensions of the main arguments. 
The model produced the main proof idea for extending the dimension-free $3$-approximation from $p$-norm social costs to monotone symmetric norm objectives (\Cref{thm:monotone-social-cost}); JH subsequently checked, edited, and incorporated the argument into the manuscript.

JH accepts full responsibility for the correctness of all material in this work.

%% file: higher-appendix.tex
In this section, we prove various lemmas regarding our defined functions $g$ and $G$.
As a reminder, the function $g$ is given by
\[
g(z) = \qnorm{z - f}^p - \lambda^p \qnorm{z}^p.
\]
The function $G$ is given by
\begin{subequations}
\begin{equation*}
G(Z) = \begin{cases}
    - Z^{p/q}  \lambda^{p} \left ( 1 - C^* \right )^{1-p} + (1 - Z)^{p/q}  & \text{if } p \geq q > 1 \\
    - Z^{p/q}  \lambda^{p} \left ( 1 - C^* \right )^{1-p} + (1 - Z)^{p/q} \left (1 - \gamma^{\frac{p}{1-p/q}} \right )^{1-p/q}   & \text{if } 1 \leq p < q
\end{cases}
\end{equation*}
where $\gamma \in (\lambda, 1)$ is a parameter that will later be specified and $C^*$ is given by
\begin{equation*}
C^* \coloneqq C^*(p, q, \gamma) = \begin{cases}
0 & \text{if } p = 1, \\
\lambda^\frac{p}{p-1} & \text{if } p \geq q > 1,  \\
\left ( \lambda / \gamma \right )^ \frac{p}{p-1} & \text{if } 1 < p < q.
\end{cases}
\end{equation*}
\end{subequations}
We conveniently write $G(Z) = \delta_1 (1-Z)^r - \delta_2 Z^r$
for some constants $\delta_1, \delta_2$ and $r = p/q > 0$.

\begin{claim}
$\delta_1 \in (0, 1]$ and $\delta_2 > 0$ for all $p \geq 1$, $q > 1$, $\lambda \in (0, 1)$ and $\gamma \in (\lambda, 1)$.
\end{claim}
\begin{proof}
If $p \geq q > 1$, we clearly have $\delta_1 = 1$.
If $1 \leq p < q$, we have $r = p/q \in (0, 1)$ and $\gamma \in (\lambda, 1)$, which implies that $\gamma^{\frac{p}{1-r}} \in (0, 1)$. Thus, $\delta_1 = (1 - \gamma^{\frac{p}{1-r}})^{1-r}$ is strictly between $0$ and $1$.
In both regimes, since $\lambda > 0$ and $C^* \in [0, 1)$, it follows that $\delta_2 = \lambda^p (1 - C^*)^{1-p} > 0$.
\end{proof}

\subsection{Proof of \texorpdfstring{\Cref{lem:g-lower-bound}}{Lower Bound for g}}\label{app:higher-d-g-lowerbound}
For a local optimum $x_i \in \mathbb R^d$, we defined the variable
\begin{equation*}
    X_i = \sum_{\substack{j \in [d], \\ x_{ij} \geq 0^+} } f_j^q.
\end{equation*}
We now restate and prove the lemma that lower bounds $g(x_i)$ in terms of $G(X_i)$.
\lemglowerbound*
\begin{proof}[Proof of \Cref{lem:g-lower-bound}]
Since $x_i$ is a local optimum, \Cref{lem:higher-d-optimal-z} implies that $x_i$ is either strictly nonpositive or strictly nonnegative. As such, we consider cases based on the entries of $x_i$.

\paragraph{Case 1: $x_i = 0$.}
If $x_i = 0$, then $g(x_i) = \qnorm{-f}^p = 1$. Since $X_i = \sum_{x_{ij} \geq 0^+} f_j^q \in [0, 1]$, and $\delta_2 > 0$, we have $G(X_i) = \delta_1 (1 - X_i)^r - \delta_2 X_i^r \leq \delta_1 (1 - X_i)^r \leq \delta_1$. Since we previously established that $\delta_1 \in (0, 1]$, we conclude that $g(x_i) = 1 \geq \delta_1 \geq G(X_i)$, satisfying the lemma.

\paragraph{Case 2: $x_{ij} < 0^-$ for some $j \in [d]$.}
By \Cref{lem:higher-d-optimal-z}, all the entries of $x_i$ are nonpositive. By \Cref{clm:wlog-points}, we can assume that there is no index $k$ such that $f_k > 0$ and $x_{ik} = 0^+$. Thus, we can derive that
\[
X_i = \sum_{\substack{j \in [d] \\ x_{ij} \geq 0^+}} f_j^q
= \sum_{\substack{j \in [d] \\ x_{ij} = 0^+}} f_j^q + \sum_{\substack{j \in [d] \\ x_{ij} > 0^+}} f_j^q = 0.
\]
Thus, we have $G(X_i) = \delta_1 (1 - 0)^r - \delta_2\cdot 0^r = \delta_1$. We now aim to lower bound $g(x_i)$. Note that
\begin{align*}
g(x_i)
&= \qnorm{x_i - f}^{p} - \lambda^p \qnorm{x_i}^p \\
&= \left ( \sum_{j \in [d]}  (|x_{ij}| + f_j )^q \right )^{p/q}  -\lambda^p \qnorm{x_i}^p \\
&\geq \left ( \sum_{j \in [d]}  |x_{ij}|^q + f_j^q \right )^{p/q} -\lambda^p \qnorm{x_i}^p \\
&= \left ( \qnorm{x_i}^q + 1 \right )^{p/q} -\lambda^p \qnorm{x_i}^p.
\end{align*}

We now make use of the following claim to complete our analysis.
\begin{claim}
Let $h(w) = (w + 1)^r - \lambda^p w^r$.
For every $w \geq 0$, we have $h(w) \geq G(0) = \delta_1$.
\end{claim}
\begin{proof}
If $r \geq 1$ and $w \geq 0$, we can apply superadditivity:
\begin{align*}
h(w) &= (1+w)^r - \lambda^p w^r \\
&\geq 1 + w^r - \lambda^p w^r \\
&= 1 + (1 - \lambda^p) w^r \\
&\geq 1.
\end{align*}
In the final line we used the fact that $\lambda \in (0, 1)$ and $w \geq 0$.
Recall that in the $r \geq 1$ regime, we explicitly defined $\delta_1 = 1$. Thus, $h(w) \geq 1 = \delta_1$ as desired.

If $0 < r < 1$, we can lower bound $h(w)$ using Jensen's inequality. For any $\eta \in (0, 1)$, Jensen's inequality implies that $(1+w)^r \geq (1-\eta)^{1-r} + \eta^{1-r} w^r$. Thus, we can lower bound $h(w)$ as follows:
\begin{align*}
h(w) &= (1+w)^r - \lambda^p w^r \\
&\geq (1-\eta)^{1-r} + \eta^{1-r} w^r - \lambda^p w^r \\
&= (1-\eta)^{1-r} + (\eta^{1-r} - \lambda^p) w^r.
\end{align*}
Recall that in the $r \in (0, 1)$ regime, we explicitly defined $\delta_1 = (1 - \gamma^{\frac{p}{1-r}})^{1-r}$. We choose $\eta = \gamma^{\frac{p}{1-r}}$, which lies strictly in $(0, 1)$ since $\gamma \in (\lambda, 1)$. Substituting this choice of $\eta$ yields:
\begin{equation*}
h(w) \geq \delta_1 + (\gamma^p - \lambda^p) w^r.
\end{equation*}
Because we strictly enforce $\gamma > \lambda$, we have $\gamma^p - \lambda^p > 0$. Since $w \geq 0$, the term $(\gamma^p - \lambda^p) w^r$ is non-negative, meaning $h(w) \geq \delta_1$ as desired.
\end{proof}

By the claim, $g(x_i) \geq h(\qnorm{x_i}^q) \geq \delta_1 = G(X_i)$, satisfying the lemma.

\paragraph{Case 3: $x_{ij} > 0^+$ for some $j \in [d]$.}
By \Cref{lem:higher-d-optimal-z}, all the entries of $x_i$ are nonnegative, and for every $j \in [d]$ with $x_{ij} > 0^+$, we have $(1 -c(x_i)) \cdot x_{ij} = f_j$, where $c(x_i) \in [0, 1)$. By \Cref{clm:wlog-points}, we can also assume that whenever $f_{j} > 0$ and $x_{ij} \geq 0^+$, we have $x_{ij} \geq f_j$. Note that
\begin{align*}
\qnorm{x_i - f}^p
&= \left ( \sum_{\substack{j \in [d] \\ x_{ij} > 0^+}} \left| \frac{f_j}{1 - c(x_i)} - f_j \right|^q + \sum_{\substack{j \in [d] \\ x_{ij} = 0}} f_j^q  \right )^{p/q}
=  \left ( \left(\frac{c(x_i)}{1-c(x_i)}\right)^q X_i + (1 - X_i) \right )^{p/q}, \\
\qnorm{x_i}^p &= \left ( \sum_{\substack{j \in [d] \\ x_{ij} > 0^+}} \left | \frac{f_j}{1 - c(x_i)} \right|^q + \sum_{\substack{j \in [d] \\ x_{ij} = 0}} |x_{ij}| \right )^{p/q}
= \left (\frac{1}{1-c(x_i) } \right )^p X_i^{p/q}.
\end{align*}
Thus, we can rewrite $g(x_i)$ as follows:
\begin{equation}\label{eq:higher-d-g-term}
    g(x_i) = \left ( \left(\frac{c(x_i)}{1-c(x_i)}\right)^q X_i+ (1 - X_i) \right )^{r} - \left (\frac{\lambda}{1 - c(x_i)} \right )^p X_i^{r}.
\end{equation}

We note that \citet{Gravin:2025aa} are able to algebraically isolate the term $c(x_i)/ (1 - c(x_i))$ in \eqref{eq:higher-d-g-term} when $p= 1$ and derive an explicit closed form in terms of $\lambda$ and $q$. Outside of a few special settings of $p$ and $q$, we are unable to do the same.
However, we can establish lower bounds on $g(x_i)$ that are independent of the scalar $c(x_i)$ by making use of the following claim.

\begin{claim}\label{clm:C-optimization}
For every $c \in [0, 1)$ and every $\gamma \in (\lambda, 1]$, we have
\begin{equation}
\frac{\gamma^p c^{p} - \lambda^p }{(1 - c)^p} \geq -\frac{\lambda^p}{(1 - c^*)^{p-1}},
\end{equation}
where $c^* \in [0, 1)$ is defined as follows:
\begin{equation*}
    c^* = \begin{cases}
    0 & \text{if } p = 1, \\
    \left(\frac{\lambda}{\gamma}\right)^{\frac{p}{p-1}} & \text{if } p > 1.
    \end{cases}
\end{equation*}
\end{claim}
\begin{proof}
Define the function
\begin{equation*}
    h(c) = \frac{\gamma^p c^{p} - \lambda^p}{(1 - c)^p}.
\end{equation*}
We minimize this function over $c \in [0, 1)$.
Taking the first derivative with respect to $c$, we get
\begin{align*}
h'(c) &= \frac{ \left[ \gamma^p p c^{p-1} \right] (1 - c)^p - \left[ \gamma^p c^{p} - \lambda^p \right] \left[ -p ( 1 - c)^{p-1} \right] }{(1 - c)^{2p}} \\
      &= \frac{ p \left[ \gamma^p c^{p-1}(1 - c) + (\gamma^p c^p - \lambda^p) \right] }{(1 - c)^{p+1}} \\
      &= \frac{ p (\gamma^p c^{p-1} - \lambda^p) }{(1 - c)^{p+1}}.
\end{align*}
Our analysis proceeds in two cases based on the value of $p$.

\paragraph{Case 3a: $p = 1$.}
The first derivative simplifies to $h'(c) = \frac{\gamma - \lambda}{(1 - c)^2} > 0$ since $\gamma > \lambda$ and $c \in [0, 1)$. Since $h'(c) > 0$, the function $h$ is strictly increasing. Therefore when restricted to the range $[0, 1)$, the function $h$ is minimized at $c^* = 0$. In particular, we have $h(c) \geq h(c^*) = -\frac{\lambda^p}{(1 - c^*)^{p-1}}$ for every $c \in [0, 1)$, as desired.

\paragraph{Case 3b: $p > 1$.}
Since $h'(c) = 0$ when $\gamma^p c^{p-1} = \lambda^p$, the function $h$ has a critical point at $c^* = \left(\lambda / \gamma\right)^{\frac{p}{p - 1}}$. Note that $c^* \in [0, 1)$ since $\gamma > \lambda > 0$.
To verify that this critical point is a global minimum, we observe the sign of $h'(c)$ around $c^*$.
For all $c \in [0, 1)$, the denominator $(1 - c)^{p+1}$ is positive. In the numerator, since $p > 1$, the term $\gamma^p c^{p-1}$ is strictly increasing with respect to $c$. Therefore, we deduce that:
\begin{itemize}
    \item For $c < c^*$, we have $\gamma^p c^{p-1} < \lambda^p$, making $h'(c) < 0$.
    \item For $c > c^*$, we have $\gamma^p c^{p-1} > \lambda^p$, making $h'(c) > 0$.
\end{itemize}
Consequently, $h'(c)$ is strictly negative on $[0, c^*)$ and strictly positive on $(c^*, 1)$. By the first derivative test, $h(c)$ strictly decreases up to $c^*$ and strictly increases thereafter, establishing $c^*$ as the unique global minimum of $h(c)$ on $[0, 1)$.

To find this minimum value, we evaluate $h(c^*)$. Rather than substituting $c^*$ directly, we use the critical point condition $\gamma^p (c^*)^{p-1} = \lambda^p$. Multiplying both sides by $c^*$ yields $\gamma^p (c^*)^p = \lambda^p c^*$.
Substituting this into the numerator of $h(c^*)$ gives
\begin{equation*}
\gamma^p (c^*)^p - \lambda^p = \lambda^p c^* - \lambda^p = -\lambda^p(1 - c^*).
\end{equation*}
We can now simplify the evaluation of $h(c^*)$ as follows:
\begin{equation*}
    h(c^*) = \frac{-\lambda^p(1 - c^*)}{(1 - c^*)^p}
           = -\frac{\lambda^p}{(1 - c^*)^{p-1}}.
\end{equation*}
Since $h(c) \geq h(c^*)$ for every $c \in [0, 1)$, this completes the proof.
\end{proof}

Using this claim, we establish a lower bound on $g(x_i)$ based on the two regimes of $r$.

\begin{claim}
\label{clm:higher-d-g-lowerbound}
The following statements hold.
\begin{enumerate}
    \item If $r \geq 1$, then
    \begin{equation}
        g(x_i) \geq (1-X_i)^r - \frac{\lambda^p}{(1 - C^*)^{p-1}} X_i^r.
    \end{equation}
    \item For any $\gamma \in (\lambda, 1)$. If $0 < r < 1$, then
    \begin{equation}
        g(x_i) \geq (1 - \gamma^{\frac{p}{1-r}})^{1-r} (1-X_i)^r - \frac{\lambda^p}{(1 - C^*)^{p-1}} X_i^r.
    \end{equation}
\end{enumerate}
\end{claim}

\begin{proof}
Note that for any $A, B \geq 0$, we have
\begin{equation*}
(A + B)^{r} \geq
\begin{cases}
A^r + B^r & \text{if }r \geq 1, \\
\eta \left (\frac{A}{\eta} \right)^r + \left (1- \eta \right) \left (\frac{B}{1 - \eta}\right)^r & \text{if } 0 < r <  1.
\end{cases}
\end{equation*}
The second case follows from Jensen's inequality for any $\eta \in (0, 1)$.
We apply the above inequalities to lower bound the first term of $g(x_i)$ as defined in \eqref{eq:higher-d-g-term}. Recall that
\[
    g(x_i) = \left ( \left(\frac{c(x_i)}{1-c(x_i)}\right)^q X_i+ (1 - X_i) \right )^{r} - \left (\frac{\lambda}{1 - c(x_i)} \right )^p X_i^{r}.
\]
If $r \geq 1$, we have
\begin{align*}
g(x_i)
&\geq \left ( \left(\frac{c(x_i)}{1-c(x_i)}\right)^{p} X_i^r+ (1 - X_i)^r \right ) - \left (\frac{\lambda}{1 - c(x_i)} \right )^p X_i^{r} \\
&=  (1 - X_i)^r +  \left ( \frac{ c(x_i)^{p} - \lambda^p }{(1 - c(x_i))^p}  \right ) X_i^r \\
&\geq (1 - X_i)^r - \frac{\lambda^p}{(1 - c^*)^{p-1}} X_i^r \tag{by \Cref{clm:C-optimization}}.
\end{align*}
Here we applied \Cref{clm:C-optimization} with $c = c(x_i)$ and $\gamma = 1$. Under these parameters, the local minimum is $c^* = \lambda^{\frac{p}{p-1}}$, which perfectly coincides with the global parameter $C^*$ in the $r \geq 1$ regime.

If $0 < r < 1$, let $\eta = \gamma^{\frac{p}{1-r}}$. We have
\begin{align*}
g(x_i)
&\geq \left ( \eta^{1-r} \left(\frac{c(x_i)}{1-c(x_i)}\right)^{p} X_i^r+ (1 - \eta)^{1-r}(1 - X_i)^r \right ) - \left (\frac{\lambda}{1 - c(x_i)} \right )^p X_i^{r} \\
&= ( 1- \eta)^{1-r} (1 - X_i)^r +  \left ( \frac{ \eta^{1-r} c(x_i)^{p} - \lambda^p }{(1 - c(x_i))^p}  \right ) X_i^r \\
&= ( 1- \gamma^{\frac{p}{1-r}})^{1-r} (1 - X_i)^r +  \left ( \frac{ \gamma^p c(x_i)^{p} - \lambda^p }{(1 - c(x_i))^p}  \right ) X_i^r \\
&\geq ( 1- \gamma^{\frac{p}{1-r}})^{1-r} (1 - X_i)^r - \frac{\lambda^p}{(1 - c^*)^{p-1}} X_i^r \tag{by \Cref{clm:C-optimization}}.
\end{align*}
Here we applied \Cref{clm:C-optimization} with $c = c(x_i)$ and parameter $\gamma$. Under these parameters, the local minimum is $c^* = (\lambda/\gamma)^{\frac{p}{p-1}}$, which perfectly coincides with the global parameter $C^*$ in the $0 < r < 1$ regime, yielding the desired bound.
\end{proof}

By evaluating the bounds established in \Cref{clm:higher-d-g-lowerbound}, we immediately see that they perfectly match $\delta_1(1-X_i)^r - \delta_2 X_i^r = G(X_i)$. Thus, we conclude that $g(x_i) \geq G(X_i)$, satisfying the lemma.
\end{proof}

\subsection{Proof of \texorpdfstring{\Cref{lem:higher-d-second-derivative}}{Second Derivative of g}}\label{app:higher-d-second-derivative}

The following proof establishes the second derivative of $G(Z)$ and identifies its inflection point, which governs the convexity and concavity of the function across different parameter regimes.

\begin{proof}[Proof of \Cref{lem:higher-d-second-derivative}]
We analyze the second derivative of $G:[0,1] \to \mathbb R$ as defined in \eqref{eq:G-and-C-star-def}. Recall that
\[
G(Z) = \delta_1 ( 1 - Z)^r - \delta_2 Z^r,
\]
for some positive constants $\delta_1, \delta_2$.
For $r \notin \{1, 2\}$, differentiating $G(Z)$ twice yields
\[
G''(Z) = r(r-1) T(Z),
\]
where we define the auxiliary function
\[
T(Z) = \delta_1 (1 - Z)^{r-2} - \delta_2 Z^{r-2}.
\]
Setting $G''(Z) = 0$ for $Z \in (0, 1)$, we require that $\delta_1 (1 - Z)^{r-2} = \delta_2 Z^{r-2}$, which in turn implies that
$\frac{1-Z}{Z} = \left(\frac{\delta_2}{\delta_1}\right)^{\frac{1}{r-2}}$.
Solving for $Z$ yields a unique root
\[
    Z^* = \left ( 1 + \left(\frac{\delta_2}{\delta_1}\right)^{\frac{1}{r-2}} \right)^{-1},
\]
which strictly lies in $(0,1)$ because $\delta_1, \delta_2 > 0$. By definition, $G''(Z^*) = 0$.

To determine the sign of $G''(Z)$ on the remaining intervals, we determine the sign of $r(r-1)$ and $T(Z)$.
Differentiating $T(Z)$ yields
\[
T'(Z) = -(r-2)\delta_1(1-Z)^{r-3} - (r-2)\delta_2 Z^{r-3} = (2-r) \left[ \delta_1(1-Z)^{r-3} + \delta_2 Z^{r-3} \right].
\]
Because $Z \in (0,1)$ and $\delta_1, \delta_2 > 0$, the term in the square brackets is strictly positive. Therefore, the sign of $T'(Z)$ is entirely determined by the sign of $(2-r)$. We now consider cases to complete our analysis.

\paragraph{Case 1: $r \in (1, 2)$.}
If $r \in (1, 2)$, we clearly have $r(r-1) > 0$. Note also that $2-r > 0$, which implies that $T'(Z) > 0$. Thus, $T(Z)$ is strictly increasing on $(0, 1)$ with exactly one root at $T(Z^*) = 0$. Since $T(Z) < 0$ for $Z \in (0, Z^*)$ and $T(Z) > 0$ for $Z \in (Z^*, 1)$, multiplying by the positive term $r(r-1)$ preserves these signs. Therefore, $G''(Z) < 0$ for $Z \in (0, Z^*)$ and $G''(Z) > 0$ for $Z \in (Z^*, 1)$.

\paragraph{Case 2: $r \in (0, 1) \cup (2, \infty)$.}
If $r > 2$, we have $r(r-1) > 0$ but $2-r < 0$, which implies that $T'(Z) < 0$. Thus, $T(Z)$ is strictly decreasing on $(0, 1)$ with exactly one root at $T(Z^*) = 0$. Since $T(Z) > 0$ for $Z \in (0, Z^*)$ and $T(Z) < 0$ for $Z \in (Z^*, 1)$, multiplying by the positive term $r(r-1)$ preserves these signs. Therefore, $G''(Z) > 0$ for $Z \in (0, Z^*)$ and $G''(Z) < 0$ for $Z \in (Z^*, 1)$.

Similarly, if $r \in (0, 1)$, we have $r(r-1) < 0$ and $2-r > 0$, which implies that $T'(Z) > 0$. Thus, $T(Z)$ is strictly increasing on $(0, 1)$ with exactly one root at $T(Z^*) = 0$. Since $T(Z) < 0$ for $Z \in (0, Z^*)$ and $T(Z) > 0$ for $Z \in (Z^*, 1)$, multiplying by the negative term $r(r-1)$ reverses these signs. Therefore, once again, $G''(Z) > 0$ for $Z \in (0, Z^*)$ and $G''(Z) < 0$ for $Z \in (Z^*, 1)$.

\end{proof}

\subsection{Proof of  \texorpdfstring{\Cref{lem:u-b-minimization}}{Bound on Minimizer b} } \label{app:u-b-minimization}

In the case where $p/q = r \in [1, 2)$, we analyze the function $u(b) = G(b) + (2b - 1) \cdot G(0)$ to prove a lower bound on the global minimizer $b^* \in [Z^*, 1]$, where $Z^*$ is the inflection point guaranteed by \Cref{lem:higher-d-second-derivative}.

\lemubminimization*

\begin{proof}[Proof of \Cref{lem:u-b-minimization}]
The condition for the minimizer $b^* \in [Z^*, 1]$ of the function $u(b)$ is that either $b^*$ is a boundary point, meaning $b^* \in \{Z^*, 1\}$, or it is an interior point where $u'(b^*) = 0$. Since $G(0) = 1$, the derivative of $u(b)$ is given by
\begin{align*}
    u'(b) &= \frac{d}{db} \left[ G(b) + 2b - 1 \right] = G'(b) + 2.
\end{align*}

First, we note that there exists at most one solution to $u'(b) = 0$ on $b \in (Z^*, 1)$. This uniqueness is guaranteed because $u''(b) = G''(b)$, and $G(b)$ is strictly convex on $(Z^*, 1)$ since its second derivative $G''(b) > 0$ on this interval (\Cref{lem:higher-d-second-derivative}).

We now prove that if $u(b^*) = 0$, the minimizer $b^*$ must be an interior point (thus $u'(b^*) = 0$) by contradiction. Assume that the minimizer $b^*$ lies on the boundary.

\paragraph{Case 1: $b^* = 1$.}
If the minimum occurs at $b^* = 1$, then $u(1) = 0$. Evaluating $u(1)$, we find that
\begin{align*}
    u(1) = G(1) + 2(1) - 1 = (1 - 1)^r - \delta_2(1)^r + 1 = 1 - \delta_2.
\end{align*}
Thus, $u(1) = 0$ implies $\delta_2 = 1$.
If $\delta_2 = 1$, the inflection point $Z^*$ where $G''(Z^*) = 0$ satisfies
\begin{align*}
    r(r-1)(1-Z^*)^{r-2} - r(r-1)(Z^*)^{r-2} = 0
\end{align*}
Canceling like terms, we see that  $1-Z^* = Z^*$, and thus we also have  $Z^* = 1/2$.
However, if we evaluate $u(b)$ at this left boundary $Z^* = 1/2$, we get
\begin{align*}
    u(1/2) = G(1/2) + 2(1/2) - 1 = (1/2)^r - (1/2)^r = 0.
\end{align*}
This means $u(1/2) = 0$ and $u(1) = 0$. Because $u(b)$ is strictly convex on $(1/2, 1)$, a secant line drawn between two roots must lie strictly above the curve. Therefore, $u(b)$ must be strictly negative for all $b \in (1/2, 1)$. This contradicts our premise that the absolute minimum value of $u(b)$ is $0$. Thus, $b^* \neq 1$.

\paragraph{Case 2: $b^* = Z^*$.}
If the minimum occurs at $b^* = Z^*$, then $u(Z^*) = 0$.
By definition, $G''(Z^*) = 0$, which gives $\delta_2 = \left(\frac{1-Z^*}{Z^*}\right)^{r-2}$. Substituting this into $u(Z^*) = 0$, we derive that
\begin{align*}
    u(Z^*) &= (1-Z^*)^r - \left(\frac{1-Z^*}{Z^*}\right)^{r-2} (Z^*)^r + 2Z^* - 1 \\
    &= (1-Z^*)^{r-2} \left[ (1-Z^*)^2 - (Z^*)^2 \right] + 2Z^* - 1 \\
    &= (1-Z^*)^{r-2} (1-2Z^*) - (1-2Z^*) \\
    &= (1-2Z^*) \left[ (1-Z^*)^{r-2} - 1 \right] = 0.
\end{align*}
Since $Z^* \in (0, 1)$ and $r \in (1, 2)$, the exponent $r-2$ is strictly negative. Therefore, $(1-Z^*)^{r-2} > 1$, which means the term in the brackets is strictly positive. The only way this expression evaluates to $0$ is if $1 - 2Z^* = 0$, meaning $Z^* = 1/2$.
If $Z^* = 1/2$, then $\delta_2 = 1$. As established in Case 1, setting $\delta_2 = 1$ forces the strictly convex function to have roots at both boundaries ($1/2$ and $1$). This forces the function to dip below $0$ in the interior, which again contradicts the minimum being $0$.

Since $b^*$ cannot be $Z^*$ and cannot be $1$, the minimizer must strictly lie in the interior $(Z^*, 1)$. Therefore, we conclude that $u'(b^*) = 0$.
\end{proof}

\subsection{Proof of \texorpdfstring{\Cref{lem:u-a-minimization-r-in-2-inf}}{Bound on Minimizer a} and  \texorpdfstring{\Cref{lem:u-a-minimization-r-in-zero-one}}{Bound on Minimizer a}} \label{app:u-a-minimization}

In the case of $p/q = r \in (0, 1) \cup (2, \infty)$, we analyze the function
\begin{equation*}
u(a) = G(a) + (1-2a) \cdot G(1),
\end{equation*}
where $a^* \in \arg\min_{a \in [0, \min(Z^*, 1/2)]} u(a)$ is a global minimizer and $Z^*$ is the unique inflection point guaranteed by \Cref{lem:higher-d-second-derivative}. Since \Cref{lem:u-a-minimization-r-in-2-inf} and \Cref{lem:u-a-minimization-r-in-zero-one} make the exact same claim for their respective regimes, we restate and prove them together as a single lemma:

\begin{lemma}[Bound on Minimizer $a$]
For $r \in (0, 1) \cup (2, \infty)$, if $u(a^*) = 0$, then $u'(a^*) = 0$ and $a^* < 1/2$.
\end{lemma}

\begin{proof}[Proof of \Cref{lem:u-a-minimization-r-in-2-inf} and \Cref{lem:u-a-minimization-r-in-zero-one}]
The condition for the minimizer $a^* \in [0, \min(Z^*, 1/2)]$ of the function $u(a)$ is that either $a^*$ is a boundary point, meaning $a^* \in \{0, \min(Z^*, 1/2)\}$, or it is an interior point where $u'(a^*) = 0$. Since we are minimizing $u(a) = G(a) + (1 - 2a) \cdot G(1)$, the derivative of $u(a)$ is given by
\begin{align*}
    u'(a) &= \frac{d}{da} \left[ G(a) + (1 - 2a) \cdot G(1) \right] = G'(a) - 2 \cdot G(1).
\end{align*}

First, we note that there exists at most one solution to $u'(a) = 0$ on $a \in (0, \min(Z^*, 1/2))$. This uniqueness is guaranteed because $u''(a) = G''(a)$, and $G(a)$ is strictly convex on $(0, \min(Z^*, 1/2))$ since its second derivative $G''(a) > 0$ on this interval (\Cref{lem:higher-d-second-derivative}).

We now prove that if $u(a^*) = 0$, the minimizer $a^*$ must be an interior point (thus $u'(a^*) = 0$) by contradiction. Assume that the minimizer $a^*$ lies on the boundary.

\paragraph{Case 1: $a^* = 0$.}
If the minimum occurs at $a^* = 0$, then $u(0) = 0$. For $0$ to be a local minimum on $[0, \min(Z^*, 1/2)]$, the function cannot be decreasing as it moves into the interval, so we assume $u'(0) > 0$ (if $u'(0) = 0$, the lemma is immediately satisfied).
Evaluating $u(0)$, we find that
\begin{align*}
    u(0) = G(0) + (1-0) \cdot G(1) = G(0) + G(1) = 0.
\end{align*}
Thus, $u(0) = 0$ implies $G(1) = -G(0)$. Substituting this into our assumption that $u'(0) > 0$, we derive that
\begin{align*}
    0 &< G'(0) - 2 \cdot G(1) \\
      &= G'(0) - 2 \cdot (-G(0)) \\
      &= G'(0) + 2 \cdot G(0).
\end{align*}
This directly implies that $G'(0) + 2 \cdot G(0) > 0$.

However, recall that $G(0) = \delta_1 (1-0)^r - \delta_2 (0)^r = \delta_1 > 0$. On the other hand, we know that the first derivative of $G$ is
\begin{align*}
    G'(a) = -r\delta_1(1-a)^{r-1} - r\delta_2 a^{r-1}.
\end{align*}
Because $\delta_1, \delta_2 > 0$ and $r \in (0, 1) \cup (2, \infty)$, we analyze $G'(0) + 2 \cdot G(0)$ by cases:
\begin{itemize}
    \item If $r \in (2, \infty)$, evaluating at $a=0$ gives $G'(0) = -r\delta_1$. Then $G'(0) + 2 \cdot G(0) = -r\delta_1 + 2\delta_1 = \delta_1(2 - r)$. Since $r > 2$ and $\delta_1 > 0$, we have $\delta_1(2 - r) < 0$, which strictly contradicts the derived requirement.
    \item If $r \in (0, 1)$, as $a \to 0^+$, the term $-r\delta_2 a^{r-1} \to -\infty$. Thus $G'(a) \to -\infty$, making $G'(0) + 2 \cdot G(0) < 0$ as well, which again strictly contradicts the derived requirement that $G'(0) + 2 \cdot G(0) > 0$.
\end{itemize}
Thus, $a^* \neq 0$.

\paragraph{Case 2: The right boundary.}
If the right boundary is $Z^*<1/2$ and the minimum occurs at $a^* = Z^*$, then $u(Z^*) = 0$.
By definition, $G''(Z^*)=0$, so $\delta_2=\delta_1\left((1-Z^*)/Z^*\right)^{r-2}$.
Substituting this into $u(Z^*)=G(Z^*)+(1-2Z^*)G(1)=0$ gives
\begin{align*}
0
&= \delta_1(1-Z^*)^r
    - \delta_2 (Z^*)^r
    - (1-2Z^*)\delta_2 \\
&= \delta_1(1-Z^*)^{r-2}
    \left( (1-2Z^*)\left[1-(Z^*)^{2-r}\right]\right).
\end{align*}
Since $\delta_1(1-Z^*)^{r-2}>0$ and $1-2Z^*>0$, the equality above would force $1-(Z^*)^{2-r}=0$, or equivalently $(Z^*)^{2-r}=1$. This is impossible because $Z^*\in(0,1)$: if $r>2$, then $(Z^*)^{2-r}>1$, while if $r\in(0,1)$, then $(Z^*)^{2-r}<1$.
Thus $a^*\neq Z^*$ when $Z^*<1/2$.

If the right boundary is $1/2$ and the minimum occurs at $a^* = 1/2$, then $u(1/2) = 0$. This implies $G(1/2) = 0$, meaning that $\delta_2 = \delta_1$. We can then evaluate the derivative at $1/2$ to find that $u'(1/2) = G'(1/2) + 2\delta_1 = \delta_1(2 - r2^{2-r})$. Because $r2^{2-r} < 2$ for all $r \in (0, 1) \cup (2, \infty)$, we have $u'(1/2) > 0$. Since $u(1/2) = 0$ and $u(a)$ is strictly increasing at $1/2$, there must exist $a < 1/2$ where $u(a) < 0$. This contradicts $a^*$ being the global minimizer with a value of $0$. Thus, $a^* \neq 1/2$.

Since $a^*$ cannot be $\min(Z^*, 1/2)$ and cannot be $0$, the minimizer must strictly lie in the interior $(0, \min(Z^*, 1/2))$. Therefore, we conclude that $u'(a^*) = 0$.
\end{proof}